\newif\iftwocolumn
\newif\ifonecolumn
\newif\iflncs
\newif\ifccs
\newif\ifanonymous
\newif\iftoday
\todaytrue
\lncstrue

\ifccs
  \twocolumntrue
\fi
\iflncs
  \onecolumntrue
\fi

\iflncs
\documentclass[runningheads]{llncs}
\fi
\ifccs
\documentclass[sigconf, anonymous]{acmart}
\input{ccs-preamble}
\fi

\usepackage{preamble}
\usepackage{geometry}
\begin{document}
\title{
  \podtm: An Optimal-Latency, Censorship-Free, \\and Accountable Generalized Consensus Layer
}
\ifanonymous{\iflncs
\fi}
\else
\author{
  Orestis Alpos\inst{1}\and
  Bernardo David\inst{1,2}\and
  Jakov Mitrovski\inst{1,3}\and\\
  Odysseas Sofikitis\inst{1,4}\and
  Dionysis Zindros\inst{1,4}
}
\iflncs
\institute{
  Common Prefix\and
  IT University of Copenhagen (ITU)\and
  Technical University of Munich\and
  \podtm network
}
\else
\affiliation{
\institution{
}
}
\fi
\fi

\iflncs
\maketitle
\fi

\iftoday
\noindent
\fi

\begin{abstract}
This work addresses the inherent issues of high latency in blockchains and low scalability in traditional consensus protocols.
We present \podtm, a novel notion of consensus whose first priority is to achieve the physically-optimal latency of $2\delta$, or one round-trip, \textit{i.e.}, requiring only one network trip (duration $\delta$) for writing a transaction and one for reading it.

To accomplish this, we first eliminate inter-replica communication. Instead, clients send transactions directly to all replicas, which independently process transactions and append them to local logs.
Replicas assigns a timestamp and a sequence number to each transaction in their logs, allowing clients to extract valuable metadata about the transactions and the system state. Later on, clients retrieve these logs and extract transactions (and associated metadata) from them.

Necessarily, this construction achieves weaker properties than a total-order broadcast protocol, due to existing lower bounds.
Our work models the primitive of \podtm and defines its security properties.
We then show \podtm-core, a protocol that satisfies properties such as transaction confirmation within $2\delta$, censorship resistance against Byzantine replicas, and accountability for safety violations. We show that single-shot auctions can be realized using the \podtm notion and observe that it is also sufficient for other popular applications.

\end{abstract}

\ifccs
\input{ccs-keywords}
\maketitle
\fi

\section{Introduction}
Despite the widespread adoption of blockchains, a significant challenge remains unresolved: they are inherently slow. The latency from the moment a client submits a transaction to when it is confirmed in another client's view of the blockchain can be prohibitively long for certain applications. Notice that we define latency in terms of the blockchain \emph{liveness} property, referring to finalized, non-reversible outputs: once a transaction is received by a reader, it remains in the protocol's output permanently.
Moreover, we do not assume ``optimistic'' or ``happy path'' scenarios, where transactions might finalize faster under favorable conditions (such as having honest leaders or optimal network conditions).

Indeed, Nakamoto-style blockchain protocols require a large number of rounds in order to achieve consensus on a new block, even when considering the best known bounds~\cite{DBLP:conf/crypto/GaziRR23}. On the other hand, it is known that permissioned protocols for $n$ parties (out of which $t$ are corrupted) realizing traditional notions of broadcast and Byzantine agreement require at least $t+1$ rounds in the synchronous case~\cite{DBLP:journals/ipl/AguileraT99} and at least $2n/(n-t)$ rounds in the asynchronous case~\cite{DBLP:conf/focs/GarayKKO07}, even when allowing for digital signatures and probabilistic termination.

In a model where \emph{replicas} maintain the network, \emph{writers} submit transactions, and \emph{readers} read the network, the minimum latency is one network round trip, or $2\delta$, letting $\delta$ denote the actual network delay, as the information must travel from the writers to the replicas and then to the readers. More importantly, we want that any transaction from an honest writer appears in the output of honest readers within $2\delta$ time, regardless of the current value of $\delta$ and corrupted parties' actions.
In this context, this work is motivated by the following question.


\begin{center}
	\emph{Can we realize tasks that blockchains are commonly used for with optimal latency?}
\end{center}

We give a positive answer to this question with a protocol realizing \podtm, a new notion of consensus that trades off traditional agreement properties for optimal latency, while retaining sufficient security guarantees to realize important tasks (\textit{e.g.}, decentralized auctions).

\subsection{Our Contributions}
In order to motivate the notion of \podtm, we first introduce the architecture of our protocol, \podtm-core, which realizes this notion. To achieve the single-round-trip latency, our first key design decision is to eliminate inter-replica communication entirely. Instead, writers send their transactions directly to all replicas. Each replica maintains its own \emph{replica log}, processes incoming transactions independently, and transmits its log to readers on request. Readers then process these replica logs to extract transactions and relevant associated information. See Figure~\ref{fig:architecture} for a summary of the \podtm-core architecture.

\begin{figure}[h]
    \centering
    \includegraphics[width=0.27\textwidth]{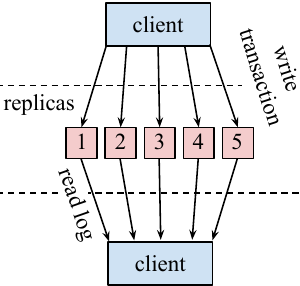}
    \caption{\podtm-core's simple architecture. A writing client (top) sends a transaction to all replicas (middle). Each replica appends it to its own log and transmits it to the reading client (bottom).}
    \label{fig:architecture}
\end{figure}

This design raises two important questions. First, what meaningful information can readers derive from replica logs when replicas operate in isolation? Second, given that in two rounds even randomized authenticated broadcast is proven impossible~\cite{DBLP:conf/focs/GarayKKO07}, what capabilities can this -- necessarily weaker -- primitive offer? We demonstrate that, by incorporating simple mechanisms, such as assigning timestamps and sequence numbers to transactions, replicas can enable readers to extract valuable information beyond mere low-latency guarantees. Furthermore, we show how the properties of \podtm can enable various applications, including auctions (as shown in \Cref{sec:auctions}).

Specifically, a secure \podtm delivers the following guarantees (formally defined in \Cref{sec:pod-model}):
\begin{itemize}
    \item Transaction confirmation within $2\delta$, with each transaction assigned a \emph{\confirmed round}: we say that the transaction becomes \emph{confirmed} at the time indicated by the \emph{confirmed round}.
    \item Censorship resistance when facing up to $\beta$ Byzantine and $\gamma$ omission-faulty replicas, ensuring all confirmed transactions appear in every honest reader's output.
    \item A \emph{past-perfect round} can be computed by readers, such that the reader is guaranteed to have received all transactions that are or will be confirmed prior to this round, even though not all transactions are strictly ordered.
    \item Accountability for all safety violations, \textit{i.e.}, if any safety property is violated, at least $\beta + 1$ replicas can be identified as misbehaving.
\end{itemize}

In particular, our Protocol \podtm-core, presented in \Cref{sec:pod-core-const}, realizes the notion of \podtm with the properties above, supporting a continuum of two adversarial models: up to \beta Byzantine replicas and up to \gamma omission-faulty replicas, out of a total of $n > 5 \beta + 3 \gamma$ replicas.
Protocol \podtm-core requires no expensive cryptographic primitives or setup beyond digital signatures and a PKI registering replicas' public keys. We showcase \podtm-core's efficiency by means of experiments with a prototype implementation presented in \Cref{sec:evaluation}. Our experiments show that even with 1000 replicas distributed around the world, the latency achieved by our protocol is just under double (resp. about 5 times) the round-trip time between writer and reader clients with security against omission-faulty (resp. Byzantine) replicas.

\subsection{Technical Overview}
We consider that time proceeds in \emph{rounds}, and that parties (replicas and clients) know the current round, so we can express \emph{timestamps} in terms of rounds. The output of \podtm associates each transaction $\tx$ with timestamp values $\rmin \geq 0$ (minimum round), $\rmax \leq \infty$ (maximum round) and $\rconf$ (\emph{confirmed round}). We call these values the \emph{trace} of \tx, and they evolve over time.
Initially we have $\rconf=\perp$ but later we get $\rconf \neq \perp$, when a transaction is \emph{confirmed}. The protocol guarantees \emph{confirmation within $u$ rounds}, meaning that, at most $u$ rounds after $\tx$ was written, every party who reads the \podtm will see $\tx$ as confirmed with some $\rconf \neq \perp$.
The protocol also guarantees that $\rmin \leq \rconf \leq \rmax$, a property we call \emph{confirmation bounds}:
while each party reads different values $\rmin,\rmax,\rconf$ for the same $\tx$, \podtm guarantees that values read by different parties stay within these limits.

When clients \emph{read} the \podtm, they obtain a $\podtm$ data structure $\Pod = (\txSet, \rperf)$, where $\txSet$ is the set of transactions and their traces and $\rperf$ is a \emph{past-perfect} round. The \emph{past-perfection} safety property guarantees that
\txSet contains \emph{all} transactions that every other honest party will ever read with a confirmed round smaller than $\rperf$.  A \podtm also guarantees \emph{past-perfection within $w$}, meaning that $\rperf$ is at most $w$ rounds in the past.


In summary, \podtm provides \emph{past-perfection} and \emph{confirmation bounds} as safety properties, ensuring parties cannot be blindsided by transactions suddenly appearing as confirmed too far in the past, and that the different (and continuously changing) transaction timestamps stay in a certain range. The liveness properties of \emph{confirmation within $u$} and \emph{past-perfection within $w$} ensure that new transactions get confirmed within a bounded delay, and that each party's past-perfect round must be constantly progressing.

Besides introducing the notion of \podtm, we present protocol \podtm-core, which realizes this notion while requiring minimal interaction among parties and achieving optimal latency, \textit{i.e.}, optimal parameters $u=2\delta$ and $w=\delta$, where $\delta$ is the current network delay (not a delay upper bound, which we assume to be unknown). Our construction relies on a set of $n$ \emph{replicas} to maintain a \podtm data structure, which can be read by an unknown number of clients. The only communication is between each client and the replicas, not among clients nor among replicas.

Writing a transaction $\tx$ to \podtm-core only requires clients to send $\tx$ to the replicas, who each assign a timestamp $\tsp$ (their current time) and a \emph{sequence number} \sn to $\tx$ and return a signature on $(\tx,\tsp, \sn)$. When reading the \podtm, the client simply requests each replica's log of transactions, validates the responses, and determines $\rmin$ and $\rmax$ from the received timestamps. If the client receives responses from enough replicas, $\rconf$ is determined by taking the median of the timestamps received from these replicas.

Protocol \podtm-core supports a continuum of mixed adversarial models, tolerating up to \beta Byzantine \emph{and} at the same time up to \gamma additional omission-faulty replicas.

\subsubsection{Applications.}
The efficiency of \podtm has the potential to allow for a plethora of distributed applications to be implemented with low latency. In \Cref{sec:auctions} we show how auctions can be run on top of \podtm, achieved through \bidsettm, a new primitive for collecting a set of bids in a censorship resistant manner. It is straightforward to realize single-shot open bid auctions using our \bidsettm primitive based on \podtm. We also conjecture that protocols for distributed sealed bid auctions based on public bulletin board can also be recast over this primitive. Moreover, we conjecture that consensusless payment systems, such as Fastpay~\cite{DBLP:conf/aft/BaudetDS20}, can also be easily realized over \podtm.


\subsection{Related work}\label{app:related}

\dotparagraph{Reducing latency}
Many previous works have lowered the latency of ordering transactions.
HotStuff~\cite{DBLP:conf/podc/YinMRGA19} uses three rounds of all-to-leader and leader-to-all communication pattern, which results in a latency (measuring from the moment a client submits a transaction until in appears in the output of honest replicas) of $8\delta$ in the happy path.
Jolteon~\cite{DBLP:conf/fc/GelashviliKSSX22}, Ditto~\cite{DBLP:conf/fc/GelashviliKSSX22}, and HotStuff-2~\cite{DBLP:journals/iacr/MalkhiN23} are two-round versions of HotStuff with end-to-end latency of $5\delta$.
MoonShot~\cite{DBLP:journals/corr/abs-2401-01791} allows leaders to send a new proposal every $\delta$ time, before receiving enough votes for the previous one, but still achieves an end-to-end latency of $5\delta$.
In the ``DAG-based'' line of word, Tusk~\cite{DBLP:conf/eurosys/DanezisKSS22} achieves and end-to-end latency of $7\delta$,
the partially-synchronous version of BullShark~\cite{DBLP:journals/corr/abs-2209-05633} an end-to-end latency of $5\delta$,
and Mysticeti~\cite{DBLP:journals/corr/abs-2310-14821} an end-to-end latency of $4\delta$.
All these protocols aim at total-order properties and have their lower latency is inherently restricted by lower bounds, whereas \podtm starts from the single-round-trip latency requirement and explores the properties that can be achieved.


\dotparagraph{Auctions}
The \podtm notion offers the \emph{past-perfection} property:
a \opread{} operation outputs a timestamp \rperf, and it is \emph{guaranteed}
that the output of \opread{} contains all transactions that can ever be confirmed with a timestamp smaller than \rperf in the view of any reading client, regardless of the network conditions.
This implies that reading clients (such as an auctioneer) cannot claim not having received a transaction when reading the \podtm, as this is detectable by any other client who reads the \podtm.
To the best of our knowledge, previous work in the consensusless literature has not considered or achieved this property,
hence it cannot readily support auctions.

\dotparagraph{Consensusless payments}
The redundancy of consensus for implementing payment systems has been recognized by previous works~\cite{DBLP:journals/dc/GuerraouiKMPS22,DBLP:journals/corr/abs-1909-10926,DBLP:conf/dsn/CollinsGKKMPPST20,DBLP:conf/aft/BaudetDS20}.
The insight is that total transaction order is not required in the case that each account is controlled by one client.
Instead, a partial order is sufficient, ensuring that, if transactions $\tx_1$ and $\tx_2$ are created by the same client, then every party outputs them in the same order.
This requirement was first formalized by Guerraoui~\etal~\cite{DBLP:journals/dc/GuerraouiKMPS22} as the \emph{source-order property}.
The constructions of Guerraoui~\etal~\cite{DBLP:journals/dc/GuerraouiKMPS22} and FastPay~\cite{DBLP:conf/aft/BaudetDS20} require clients to maintain sequence numbers.
ABC~\cite{DBLP:journals/corr/abs-1909-10926} requires clients to reference all previous transaction in a DAG (including its own last transaction).
Cheating clients might lose liveness~\cite{DBLP:conf/aft/BaudetDS20,DBLP:journals/dc/GuerraouiKMPS22,DBLP:journals/corr/abs-1909-10926}, but equivocating is not possible.

\section{Preliminaries}

\dotparagraph{Notation}
We denote by \BN the set of natural numbers including $0$.
Let $L$ be a sequence, we denote by $L[i]$ the $i^\text{th}$ element (starting from $0$),
and by $|L|$ its length.
Negative indices address elements from the end, so $L[-i]$ is the $i^\text{th}$ element from
the end, and $L[-1]$ in particular is the last.
The notation $L[i{:}]$ means the
subarray of $L$ from $i$ onwards, while $L[{:}j]$ means the subsequence of $L$ up to (but not including) $j$.
We denote an empty sequence by $[\,]$.
We denote the concatenation of sequences $L_1$ and $L_2$ by $L_1 \concat L_2$.

\subsection{Execution Model}\label{sec:system}

\dotparagraph{\textbf{Parties}}
We consider $n$ \emph{replicas} $\replicas=\{\replicas_1,\ldots,\replicas_n\}$ and an unknown number of \emph{clients}. Parties are \emph{stateful}, \textit{i.e.}, store \emph{state} between executions of different algorithms.
We assume that replicas are known to all parties and register their public keys (for which they have corresponding secret keys) in a Public Key Infrastructure (PKI). Clients do not register keys in the PKI.

\dotparagraph{\textbf{Adversarial Model}}
We call a party (replica or client) \emph{honest}, if it follows the protocol, and \emph{malicious} otherwise.
We assume \emph{static corruptions},  \textit{i.e.}, the set of malicious replicas is decided before the execution starts and remains constant.
This work uses a combination of two adversarial models, the \emph{Byzantine} and the \emph{omission} models.
In the \emph{Byzantine} model, corrupted replicas are malicious and may deviate arbitrarily from the protocol. The adversary has access to the internal state and secret keys of all corrupted parties. We denote by $\beta \in \BN$ the  number of Byzantine replicas in an execution.
The Byzantine adversary is modelled as a probabilistic polynomial time overarching entity that is invoked in the stead of every corrupted party.
That is, whenever the turn of a corrupted party comes to be invoked by the environment, the adversary is invoked instead.
In the \emph{omission} model, corrupted replicas may only deviate from the protocol by dropping messages that they were supposed to send, but follow the protocol otherwise. Observe that this includes crash faults, where replicas crash (\textit{i.e.} stop execution) and remain crashed until the end of the execution of an algorithm. We denote by $\gamma \in \BN$ the number of omission-faulty replicas in an execution.



\dotparagraph{\textbf{Modeling time}}
We assume that time proceeds in discrete \emph{rounds}, and parties have clocks allowing them to determine the current round.
For any two honest parties, their clocks can be at most $\phi$ rounds apart.
For simplicity, our analysis will assume \emph{synchronized clocks}, that is, $\phi = 0$.
Notice that although we assume synchronized clocks as a setup, clock synchronization can be achieved in partially synchronous networks~\cite{psync} using existing techniques~\cite{clockoverview}, also in the case where replicas gradually join the network~\cite{bootclocksync}.
By \emph{timestamp} we refer to a round number assigned to some event.

\dotparagraph{\textbf{Modeling network}}
We denote by $\delta \in \CN$ the actual delay (measured in number of rounds) it takes to deliver a message between two honest parties,
a number which is \emph{finite} but \emph{unknown} to all parties.
We denote by $\Delta \in \CN$ an upper bound on this delay, \textit{i.e.}, $\delta \leq \Delta$, which is also \emph{finite}.
In the \emph{synchronous} model, $\Delta$ is \emph{known} to all parties.
In the \emph{partially synchronous} model~\cite{psync}, $\Delta$ is \emph{unknown} but still finite, \textit{i.e.}, all messages are eventually delivered.
A protocol is called \emph{responsive} if it does not rely on knowledge of $\Delta$ and its liveness guarantees depend only on the actual network delay $\delta$.

\subsection{Cryptographic primitives}
\dotparagraph{\textbf{Digital Signatures}} We assume that replicas (and auctioneers in \bidsettm-core) authenticate their messages with digital signatures. A digital signature scheme is a triple of algorithms satisfying the EUF-CMA security~\cite{GMR88} as defined below:
		\begin{itemize}
	    	\item $\textsl{KeyGen}(1^{\kappa})$: The key generation algorithm takes as input a security parameter $\kappa$ and outputs a secret key $\sk$ and a public key $\pk$.
	        \item
	            $\opSign{\sk, m} \rightarrow \sig$:
	            The signing algorithm takes as input a private key
	            \sk and a message $m \in \zo^*$ and returns a signature $\sig$.
	        \item
	            $\opVerify{\pk, m, \sig} \rightarrow b \in \zo$:
	            The verification algorithm takes as input a public
	            key $\pk$, a message $m$, and a signature $\sig$, and outputs a bit $b\in \zo$.
	    \end{itemize}
We say \sig is a \emph{valid} signature on $m$ with respect to \pk if $\opVerify{pk, m, \sig} = 1$.

%
%


\subsection{Accountable safety}\label{sec:accountable-safety}
Taking a similar approach as Neu, Tas, and Tse~\cite[Def. 4]{FC:NeuTasTse22}, we define \emph{accountable safety} through an \emph{identification function}.

\begin{definition}[Transcript and partial transcript]    
    We define as \emph{transcript} the set of all network messages sent by all parties in an execution of a protocol.
    A \emph{partial transcript} is a subset of a transcript.
\end{definition}

\begin{definition}[\beta-Accountable safety]\label{def:accountable-safety}
    A protocol satisfies \emph{accountable safety} with \emph{resilience} \beta if its interface contains a function $\identify{\tran} \rightarrow \cheaters$,
    which takes as input a partial transcript $\tran$ and outputs a set of replicas $\cheaters \subset \replicas$, such that the 
    following conditions hold except with negligible probability.
    \begin{description}
        \item[Correctness:] If safety is violated, then there exists a partial transcript \tran, such that $\identify{\tran} \rightarrow \cheaters$ and $|\cheaters| > \beta$.
        \item[No-framing:] For any partial transcript \tran produced during an execution of the protocol, the output of $\identify{\tran}$ does not contain honest replicas.
    \end{description}
\end{definition}

\begin{remark}
    For the sake of simplicity, we have defined the transcript based on messages sent by \emph{all} replicas.
    We can also define a \emph{local transcript} as the set of messages observed by a single party.
    As will become evident from the implementation of \identify{}, in practice, adversarial behavior can be identified from the local transcripts of a single party or  of a pair of parties.
\end{remark}
\section{Modeling \podtm} \label{sec:pod-model}
In this section, we introduce the notion of a \podtm, a distributed protocol where clients can \emph{read} and \emph{write} transactions.
We first define basic data structures and the interface of a \podtm protocol.

\begin{definition}[Transaction trace and trace set]
    The \emph{transaction trace} of a transaction $\tx \in \zo^*$ is a tuple containing the values $(\tx, \rmin, \rmax, \rconf)$, which change during the execution of a \podtm protocol.
    We call $\rmin \in \BN$ the \emph{minimum round}, $\rmax \in \BN \cup \{\infty\}$ the \emph{maximum round},
    $\rconf \in \BN \cup \{\bot\}$ the \emph{\confirmed round}.
    We denote by $\rmax=\infty$ an unbounded maximum round and by $\rconf=\bot$ an undefined \confirmed round.
    We also denote these values as $\tx.\rmin, \tx.\rmax$, and $\tx.\rconf$.
    A \emph{trace set} $\txSet$ is a set of transaction traces $\{(\tx, \rmin, \rmax, \rconf) \mid \tx \in \zo^*\}$.
\end{definition}

\begin{definition}[Confirmed transaction]
    A transaction with confirmed round \rconf is called \emph{\confirmed} if $\rconf \neq \bot$, and \emph{unconfirmed} otherwise.
\end{definition}


\begin{definition}[Pod data structure]
    A \emph{pod data structure} $\Pod$ is a tuple $(\txSet, \rperf)$, where \txSet is a trace set and \rperf is a round number called the \emph{past-perfect round}.
\end{definition}
We denote the components of a pod data structure as $\Pod.\txSet$ and $\Pod.\rperf$.
We write $\tx \in \Pod.\txSet$ if an entry $(\tx, \cdot, \cdot, \cdot)$ exists in $\Pod.\txSet$.
We remark that transactions in \txSet may be \confirmed on unconfirmed. Moreover, \rperf will be used to define a completeness property on \txSet (the \emph{past-perfection} property of \podtm).

\begin{definition}[Auxiliary data]
    We associate with a pod data structure $\Pod$ some auxiliary data \cert, which will be used to validate \Pod.
    The exact implementation of \cert is irrelevant for the definition of \podtm;
    however, it may be helpful to mention that in \podtm-core it will be a tuple $\cert = (\ppcert, \alltxcert)$.
    $\ppcert$ will be called the \emph{past-perfection certificate} and \alltxcert will be a map from each transaction \tx in $\Pod.\txSet$ to a \emph{transaction certificate} \txcert for \tx.
    Both will contain digital signatures.
\end{definition}

\begin{definition}[Interface of a \podtm]\label{def:podinterface}
A \podtm protocol has the following interface.
\begin{itemize}
    \item $\opwrite{\tx}$:
    It writes a transaction \tx to the \podtm.

    \item $\opreadall{} \returns (\Pod, \cert)$:
    It outputs a pod data structure $\Pod = (\txSet, \rperf)$ and auxiliary data \cert.
\end{itemize}
\end{definition}
We say that a client \emph{reads the} \podtm when it calls \opread{}.
If $\tx$ appears in $\txSet$,
 we say that the client \emph{observes} \tx and, if $\tx.\rconf \neq \bot$, we say that the client \emph{observes \tx as \confirmed}.

\begin{definition}[Validity function]
    Apart from its interface functions,
    a \podtm protocol also specifies a computable, deterministic, and non-interactive function \opValid{\Pod, \cert} that takes as input a pod data structure $\Pod$ and auxiliary data $\cert$ and outputs a boolean value. We say that a pod data structure $\Pod$ is \emph{valid} if $\opValid{\Pod, \cert} = \true$.
\end{definition}


\begin{definition}[View of the \podtm]
    We call \emph{view of the} \podtm and denote by \podView{\client}{\round} the data structure returned by \opreadall{}, where \opreadall{} is invoked by client \client and the output is produced at round \round.
    We remark that \round denotes the round when \opreadall{} outputs, as the client may have invoked it at an earlier round.
\end{definition}

We now introduce the basic definition of a \emph{secure} \podtm protocol, as well as some additional properties (\emph{timeliness} and \emph{monotonicity}) that it may satisfy, which we later use for some applications.

\begin{definition}[Secure \podtm]\label{def:podsec}
A protocol is a \emph{secure} \podtm if it implements the \podtm interface of Definition~\ref{def:podinterface} and specifies a validity function \opValid{}, such that the following properties hold.
\begin{description}[leftmargin=4mm, labelindent=4mm]
    \item[(Liveness) Completeness:]
    Honest clients always output a valid pod data structure.
    That is, if \opread{} returns $(\Pod, \cert)$ to an honest client, then $\opValid{\Pod, \cert} = \true$.

    \item[(Liveness) Confirmation within \parConf:]
	Transactions of honest clients become \confirmed after at most \parConf rounds.
	Formally, if an honest client $\client$ writes a transaction \tx at round \round,
    then for any honest client $\client'$ (including $\client=\client'$) it holds that $\tx \in \podView{\client'}{\round + \parConf}$ and $\tx.\rconf \neq \bot$.

	\item[(Liveness) Past-perfection within \parPerf:]
	Rounds become past-perfect after at most $\parPerf$ rounds.
	Formally,
    for any honest client \client and round $\round \geq \parPerf$, it holds that $\podView{c}{\round}.\rperf \geq \round - \parPerf$.

	\item[(Safety) Past-perfection:]
    A valid pod \Pod contains all transactions that may ever obtain a \confirmed round smaller than $\Pod.\rperf$. Formally,
	the adversary cannot output $(\Pod_1, \cert_1)$ and $(\Pod_2, \cert_2)$ to the network, such that $\opValid{\Pod_1, \cert_1} \land \opValid{\Pod_2, \cert_2}$ and there exists a transaction $\tx$ such that $(\tx, \rmin^1, \rmax^1, \rconf^1) \not\in \Pod_1.\txSet$ and $(\tx, \rmin^2, \rmax^2, \rconf^2) \in \Pod_2.\txSet$ and $\rconf^2 \neq \bot$ and $\rconf^2 < \Pod_1.\rperf$.

	\item[(Safety) Confirmation bounds:]
    \sloppy{
    The values $\rmin$ and $\rmax$ bound the \confirmed round that a transaction may ever obtain. Formally,
    the adversary cannot output $(\Pod_1, \cert_1)$ and $(\Pod_2, \cert_2)$ to the network, such that $\opValid{\Pod_1, \cert_1} \land \opValid{\Pod_2, \cert_2}$ and there exists a transaction $\tx$ such that $(\tx, \rmin^1, \rmax^1, \rconf^1) \in \Pod_1.\txSet$ and $(\tx, \rmin^2, \rmax^2, \rconf^2) \in \Pod_2.\txSet$ and $\rmin^1 > \rconf^2$ or $\rmax^1 < \rconf^2$.
    }

\end{description}
\end{definition}

The \emph{confirmation bounds} property gives $\rmin^1 \leq \rconf^2 \leq \rmax^1$, for $\rmin^1, \rmax^1, \rconf^2$ computed by honest clients, but it does not guarantee anything about the values of $\rmin^1$ and $\rmax^1$ (for example, it could trivially be $\rmin^1=0$ and $\rmax^1=\infty$).
To this purpose we define an additional property of \podtm, called \emph{timeliness}.
Previous work has observed a similar property as orthogonal to safety and liveness~\cite{EPRINT:TziSriZin23}.

\begin{definition}[\podtm \parTemp-timeliness for honest transactions]\label{def:pod-timely}
    A \podtm protocol is \emph{\parTemp-timely} if it is a secure \podtm, as per Definition~\ref{def:podsec}, and for any honest clients $\client_1,\client_2$, if $\client_1$ writes transaction \tx in round $\round$ and $\client_2$ has view $\podView{\client_2}{\round'}$ in round $\round'$, such that $(\tx, \rmin, \rmax, \rconf) \in \podView{\client_2}{\round'}.\txSet$, then:
    \begin{enumerate}
        \item $\rconf \in (\round, \round + \parTemp]$
        \item $\rmax \in (\round, \round + \parTemp]$
        \item $\rmax - \rmin < \parTemp$, implying that $\rmin \neq 0$ and $\rmax \neq \infty$.
    \end{enumerate}
\end{definition}

Moreover, a \podtm protocol allows the values \rmin, \rmax, \rconf to change during an execution -- for example, clients in construction \podtm-core will update them when they receive votes from replicas.
The properties we have defined so far do not impose any restriction on how they evolve.
For this reason, in Appendix~\ref{app:monotonicity} we define the additional property of \podtm \emph{monotonicity}.

We conclude this section with some visual examples in Figures \ref{fig:temperature} and \ref{fig:example-output}.

\begin{figure}[ht]
    \begin{minipage}[b]{.48\textwidth}
    \centering
    \includegraphics[width=1\textwidth]{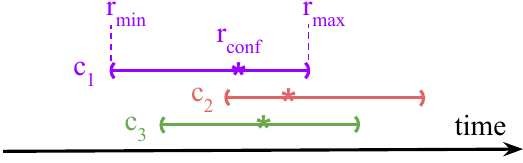}
    \caption{The same transaction in the view of three different \podtm clients. Each client assigns it a minimum round \rmin and a maximum round \rmax. If it gets confirmed, the confirmation round \rconf will be
     between these two values. The \rconf that each client locally computes respects the bounds of each other client.}
     \label{fig:temperature}
    \end{minipage}
    \hfill
    \begin{minipage}[b]{.49\textwidth}
    \centering
    \includegraphics[width=1\textwidth]{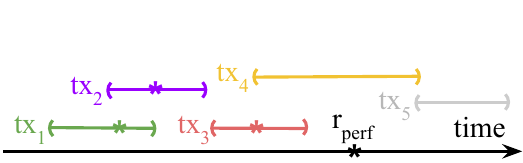}
    \caption{A possible view of a single \podtm client.
    Transactions $\tx_1, \tx_2, \tx_3$ are confirmed, $\tx_4$ is not yet confirmed.
    A client also derives a past-perfect round $\rperf$. No transaction other than $\tx_1, \tx_2, \tx_3, \tx_4$ may  obtain $\rconf \leq \rperf$. There may exist $\tx_5$ for which the client has not received votes, but $\tx_5$ cannot obtain $\rconf \leq \rperf$.
    }
    \label{fig:example-output}
    \end{minipage}
\end{figure}


\section{Protocol \podtm-core}\label{sec:pod-core-const}

Before we present protocol \podtm-core, we define basic concepts and structures.

\begin{definition}[Vote]
    A \emph{vote} is a tuple $\vote = (\tx, \tsp, \sn, \sig, \rep)$, where \tx is a transaction, \tsp is a timestamp, \sn is a sequence number, \sig is a signature, and  \rep is a replica.
    A vote is \emph{valid} if \sig is a valid signature on message $m = (\tx, \tsp, \sn)$ with respect to the public key $\pk_\rep$ of replica \rep.
\end{definition}




\begin{remark}[Processing votes in order]
    We require that clients process votes from each replica in the same order, namely in order of increasing timestamps. For this we employ \emph{sequence numbers}. Each replica maintains a sequence number, which it increments and includes every time it assigns a timestamp to a transaction.
\end{remark}
\begin{remark}[Implicit session identifiers]
We assume that all messages between clients and replicas are concatenated with a session identifier (sid), which is unique for each concurrent execution of the protocol. Moreover, the sid is implicitly included in all messages signed by the replicas.
\end{remark}
\begin{remark}[Streaming construction]
    The client protocol we show in \Cref{const:pod-core} is \emph{streaming}, that is, clients maintain a connection to the replicas, and \emph{stateful}, that is, they persist their state (received transactions and votes) across all invocations of \opwrite{} and \opreadall{}.
\end{remark}

\dotparagraph{Past-perfection and transaction certificates}
In \podtm-core, clients store certain votes which they output upon \opreadall{} as part of the \emph{certificate} \cert, which will be used to prove the validity of the returned \Pod and for accountability in case of safety violations.
Specifically, \cert consists of two parts, $\cert = (\ppcert, \alltxcert)$: the \emph{past-perfection certificate} \ppcert contains, for each replica, the vote on the most recent timestamp received from that replica. It is implemented as a map from replicas to votes, i.e., $\ppcert: \rep \rightarrow \vote$.
The \emph{transaction certificate} \alltxcert contains, for each transaction, all valid votes received for it.
It is implemented as a map from transactions to a map from replicas to votes, i.e., $\alltxcert: \tx \rightarrow \txcert$ and $\txcert: \rep \rightarrow \vote$.
We remark that $\ppcert$ can be derived by taking the union of certificates $\txcert$ for all transactions and keeping the most recent vote for each replica, but we define $\ppcert$ explicitly for clarity and readability.

\dotparagraph{Pseudocode notation}
The notation `\textbf{require} $P$' causes a function to terminate immediately and return $\false$ if $P$ evaluates to $\false$.
Notation `\textbf{upon} $e$' causes a block of code to be executed when event $e$ occurs.
Notations `$\msg{MSG}{} \leftarrow p$' and `$\msg{MSG}{} \rightarrow p$' denote receiving and sending a message $MSG$ from and to party $p$, respectively.
Finally, $x: a \in A \rightarrow b \in B$ denotes that variable $x$ is a map from elements of type $A$ to elements of type $B$. When obvious from the context, we do not explicitly write the types $A$ or $B$.
For a map $x$, the operations $x.\opKeys$ and $x.\opValues$ return all keys and all values in $x$, respectively.
With $\emptyset$ we denote an empty map.

\begin{protocol}[\podtm-core]\label{const:pod-core}
     Protocol \podtm-core is executed by $n$ replicas that follow the steps of \Cref{alg:pod-core-replica} and an unknown number of clients that follow the steps of Algorithms \ref{alg:client-1} and \ref{alg:client-statistics} with parameters \beta, \gamma and \alpha, where \beta denotes the number of Byzantine replicas and \gamma the number of omission-faulty replicas (in addition to the Byzantine) and $\alpha = n - \beta - \gamma$ is the number of honest replicas.
\end{protocol}

\subsection{Replica code}

\begin{algorithm}[ht!]
    \caption{Protocol \podtm-core: Code for a replica $\rep_i$, where \sk denotes its secret signing key.}
    \label{alg:pod-core-replica}
    \begin{algorithmic}[1]
        \State{$\allClients$ }      \Comment{\text{The set of all connected clients}}\label{line:rep-state-begin}
        \State{\nextsn}             \Comment{\text{The next sequence number to assign to votes}}
        \State{$\repLog$}           \label{line:rep-state-end} \Comment{\text{The transaction log or the replica}}
        \Statex
        \vspace{-1mm}

        \On{\opInit{}} \Comment{Called once when the replica is initialized}
            \State{$\allClients \gets \emptyset ;\; \nextsn \gets 0 ;\; \repLog \gets [ \, ]$}\label{line:rep-init-clients}
        \EndOn
        \Statex
        \vspace{-1mm}

        \On{$\msgconnect{} \leftarrow \client$} \label{line:rep-receive-connect-begin} \Comment{Called when a new client \client connects to the replica}
            \Let{\allClients}{\allClients \cup \{\client\}}
            \For{$(\tx, \tsp, \sn, \sigma) \in \repLog$}
                \State{$\msgrecord{(\tx, \tsp, \sn, \sigma, \rep_i)} \rightarrow \client$}
            \EndFor
        \EndOn\label{line:rep-receive-connect-end}
        \Statex
        \vspace{-1mm}

        \On{$\msgwrite{\tx} \leftarrow \client$} \label{line:rep-receive-write-begin} \Comment{Called when a client \client writes a transaction \tx}
            \State{\textbf{if} $\repLog[\tx] \neq \bot $ \textbf{then return}} \Comment{Ignore duplicate transactions}
            \State{\opSendVote{\tx}}
        \EndOn\label{line:rep-receive-write-end}
        \Statex

        \vspace{-1mm}
        \Function{\opSendVote{\tx}}{}
            \State{$\tsp \gets \opTime ;\; \sn \gets \nextsn ;\; \sig \gets \opSign{\sk, (\tx, \tsp, \sn)}$} \label{line:rep-assign-tsp}  \Comment{\text{\opTime{} returns the current round}}
            \Let{\repLog}{\repLog \concat (\tx, \tsp, \sn, \sig)}
            \For{$\client \in \allClients$}
                \State{$\msgrecord{(\tx, \tsp, \sn, \sig, \rep_i)} \rightarrow \client$} \label{line:send-vote}
            \EndFor
            \Let{\nextsn}{\nextsn+1} \label{line:rep-increment-sn}
        \EndFunction
        \Statex

        \vspace{-1mm}
        \On{end round} \label{line:rep-send-heartbeat-begin}                      \Comment{Executed at the end of each round}
            \Let{\tx}{\heartBeat \| \opTime}
            \State{\opSendVote{\tx}}\label{line:rep-send-vote-heartbeat}
        \EndOn\label{line:rep-send-heartbeat-end}
    \end{algorithmic}
    \end{algorithm}

The state of a replica (lines \ref{line:rep-state-begin}--\ref{line:rep-state-end} of \Cref{alg:pod-core-replica}) contains \repLog, a log implemented as a sequence of votes $(\tx, \tsp, \sn, \sigma, \rep_i)$ created by the replica, where \tsp is the timestamp assigned by the replica to \tx, \sn is a sequence number, and \sig is its signature.
When the replica receives a \msgconnect{} message from a client \client, it appends \client to its set of connected clients and sends to \client all entries in \repLog (lines \ref{line:rep-receive-connect-begin}--\ref{line:rep-receive-connect-end}).

When it receives  \msgwrite{\tx}, a replica first checks whether it has already seen \tx, in which case the message is ignored. Otherwise, it assigns \tx a timestamp \tsp equal its local round number and the next available sequence number \sn, and signs the message $(\tx, \tsp, \sn)$ (\cref{line:rep-assign-tsp}).
Honest replicas use incremental sequence numbers for each transaction, implying that a vote with a larger sequence number than a second vote will have a larger or equal timestamp than the second.
The replica appends $(\tx, \tsp, \sn, \sig)$ to \repLog, and sends it via a \msgrecord{(\tx, \tsp, \sn, \sig, \rep_i)} message to all connected clients (\cref{line:send-vote}).

\dotparagraph{Heartbeat messages}
As we will see, clients maintain a \emph{most-recent timestamp} variable $\mrt[\rep_j]$ for each replica. This is updated every time they receive a vote and is crucial for computing the past-perfect round \rperf.
To make sure that clients update $\mrt[\rep_j]$ even when $\rep_j$ does not have any new transactions in a round,
we have replicas send a vote on a dummy \heartBeat transaction the end of each round (lines~\ref{line:rep-send-heartbeat-begin}--\ref{line:rep-send-heartbeat-end}).
An obvious practical optimization is to send \heartBeat only for rounds when no other transactions were sent.
When received by a client, a \heartBeat is handled as a vote (i.e., it triggers \cref{line:receive-vote-begin} in \Cref{alg:client-1}). To avoid being considered a duplicate vote by clients (see \cref{line:check-dup-tsp} in \Cref{alg:client-1}), replicas append the round number to the \heartBeat transaction.

\subsection{Client code}

\begin{algorithm}[th!]
    \caption{Protocol \podtm-core: Code for a client, part 1}
    \label{alg:client-1}
    \begin{algorithmic}[1]
        \State{\textbf{State:}}
            \State{$\allreps = \{ \rep_1, \ldots, \rep_n \} ;\; \{\pk_1, \ldots, \pk_n \}$ } \Comment{\text{All replicas and their public keys}}\label{line:client-state-begin}
            \State{$\mrt : \rep \rightarrow \tsp$}   \Comment{The most recent timestamp returned by each replica}
            \State{$\nextsn : \rep \rightarrow \sn$}   \Comment{The next sequence number expected by each replica}
            \State{$\tsps: \tx \rightarrow ( \rep \rightarrow \tsp)$}     \Comment{Timestamp received for each tx from each replica}
            \State{$\Pod = (\allTxs, \rperf)$}                   \Comment{The \podtm observed by the client so far}
            \State{$\ppcert: \rep \rightarrow \vote$} \Comment{Past-perfection certificate: the most recent vote from each replica}
            \State{$\alltxcert: \tx \rightarrow \txcert$, where $\txcert: \rep \rightarrow \vote$} \label{line:client-state-end} \Comment{Transaction certificate: for each transaction, all votes}
        \Statex
        \vspace{-1mm}
        \On{\opInit{}} \label{line:client-init-start} \Comment{Called once when the client is initialized}
            \State{\opInitClient{}}
            \State{\textbf{for} $\rep_j \in \allreps$ \textbf{do:} $\msgconnect{} \rightarrow \rep_j$}\label{line:client-init-conn}
        \EndOn\label{line:client-init-end}
        \Statex
        \vspace{-1mm}
        \On{$\msgrecord{(\tx, \tsp, \sn, \sig, \rep_j)} \leftarrow \rep_j$} \Comment{Called when client receives vote from replica $\rep_j$}\label{line:receive-vote-begin}
            \If{\opProcessVote{$\tx, \tsp, \sn, \sig, \rep_j$}}
                \Let{\ppcert[\rep_j]}{(\tx, \tsp, \sn, \sig, \rep_j)} \label{line:upd-ppcert}   \Comment{Keep most recent vote from $\rep_j$ in \ppcert}
                \Let{\alltxcert[\tx][\rep_j]}{(\tx, \tsp, \sn, \sig, \rep_j)} \label{line:client-store-vote-txcert}  \Comment{Keep all votes for \tx in \txcert}
            \EndIf
        \EndOn\label{line:receive-vote-end}

        \Statex
        \vspace{-1mm}
        \Function{$\opwrite{\tx}$}{} \label{line:write-tx-begin} \Comment{Part of \podtm interface, used to write a new transaction}
            \State{\textbf{for} $\rep_j \in \allreps$ \textbf{do:} $\msgwrite{\tx} \rightarrow \rep_j$} \label{line:send-write-tx}
        \EndFunction\label{line:write-tx-end}
        \Statex
        \vspace{-1mm}
        \Function{$\opreadall{}$}{} \label{line:read-begin} \Comment{Part of \podtm interface, used to read all transactions}
            \Let{\txSet}{\opComputeTxSet{\tsps, \mrt}}  \Comment{Shown in \Cref{alg:client-statistics}}
            \Let{\rperf}{\opMinPossibleTsForNewTx{\mrt}}\label{line:upd-rperf-begin} \Comment{Shown in \Cref{alg:client-statistics}}
            \State{\Pod $\leftarrow$ (\txSet, \rperf) ;\; \cert $\leftarrow$ (\ppcert, \alltxcert)} \label{line:client-construct-cert}
            \return{ $(\Pod, \cert)$}\label{line:return-pod}
        \EndFunction\label{line:read-end}
        \Statex
        \vspace{-1mm}
        \Function{\opInitClient{}}{}
            \State{$\tsps\leftarrow\emptyset ;\; \alltxcert\leftarrow\emptyset ;\; \Pod = (\emptyset, 0)$} \label{line:txs-map-init}
            \State{\textbf{for} $\rep_j \in \allreps$ \textbf{do:} $\mrt[\rep_j] \leftarrow 0 ;\; \ppcert[\rep_j] \leftarrow \bot ;\; \nextsn[\rep_j] = -1 $}\label{line:client-init-mrt-ppcert}

        \EndFunction
        \Statex
        \vspace{-1mm}
        \Function{\opProcessVote{$\tx, \tsp, \sn, \sig, \rep_j$}}{} \Comment{Validate vote and update local state}
            \State{\textbf{require} \opVerify{$\pk_j, (\tx, \tsp, \sn), \sig$}} \Comment{Otherwise, vote is invalid} \label{line:client-verify-sig-begin}
            \State{\textbf{require} $\sn = \nextsn[\rep_j]$} \Comment{Otherwise, vote cannot be processed yet}\label{line:client-check-sn}
            \Let{\nextsn[\rep_j]}{\nextsn[\rep_j] + 1}
            \State{\textbf{require} $\tsp \geq \mrt[\rep_j]$} \Comment{Otherwise, $\rep_j$ has sent old timestamp} \label{line:check-old-ts}
            \Let{\mrt[\rep_j]}{\tsp} \label{line:upd-mrt}
                \State{\textbf{require} $\tsps[\tx][\rep_j] = \bot$ \textbf{or} $\tsps[\tx][\rep_j] = \tsp$ } \Comment{Otherwise, vote is duplicate from $\rep_j$ on \tx} \label{line:check-dup-tsp}
            \Let{\tsps[\tx][\rep_j]}{\tsp} \label{line:rec-tsp}
        \EndFunction
    \end{algorithmic}
    \end{algorithm}
\begin{algorithm}[ht!]
    \caption{Protocol \podtm-core: Client code, part 2. Functions to compute trace values and past-perfect round. The code is parametrized with \beta, the number of Byzantine replicas expected by the client, and \gamma, the number of omission-faulty replicas, and $\alpha = n - \beta - \gamma$ for $n$ replicas.
    }
    \label{alg:client-statistics}
    \begin{algorithmic}[1]
        \Function{$\opComputeTxSet{\tsps, \mrt}$}{}
            \State{$\allTxs \gets \emptyset$}
            \For{$\tx \in \tsps.\opKeys{}$}     \Comment{loop over all received transactions} \label{line:loop-over-txs}
                \Let{\rmin}{\opMinPossibleTs{\tsps[\tx], \mrt}}\label{line:upd-rmin-begin}
                \Let{\rmax}{\opMaxPossibleTs{\tsps[\tx]}}\label{line:upd-rmax-begin}
                \State{$\rconf \gets \bot ;\; \txVotes = [ \, ]$} \label{line:unconfirmed-tx}
                \If{$| \tsps[\tx].\opKeys{} | \geq \alpha$}\label{line:confirm-tx-begin}
                    \State{\textbf{for} $\rep_j \in \tsps[\tx].\opKeys{}$ \textbf{do: } $\txVotes \leftarrow \txVotes \concat \tsps[\tx][\rep_j]$}
                    \Let{\rconf}{\opMed{\txVotes}}\label{line:assign-median}
                \EndIf\label{line:confirm-tx-end}
                \Let{\allTxs}{\allTxs \cup \{ (\tx, \rmin, \rmax, \rconf) \} }
            \EndFor
            \return \allTxs
        \EndFunction
        \Statex
        \vspace{-1mm}
        \Function{\opMinPossibleTs{\txVotesMin, \mrt}}{} \Comment{$\txVotesMin: \rep \rightarrow \tsp$, contains timestamps on tx}
        \For{$\rep_j \in \allreps$}\label{line:assing-missing-min-begin} \Comment{$\mrt: \rep \rightarrow \tsp$, most recent tsp from each replica}
            \State{\textbf{if} $\txVotesMin[\rep_j] = \bot$ \textbf{then} $\txVotesMin \gets \txVotesMin \concat [\mrt[\rep_j]]$} \label{line:fill-vote-lts}
        \EndFor\label{line:assing-missing-min-end}
        \State{\text{sort \txVotesMin in increasing order of timestamps}}
        \Let{\txVotesMin}{ [ 0, \stackrel{\beta \text{ times}}{\ldots}, 0 ] \concat \txVotesMin} \Comment{omitted altogether if $\beta=0$}\label{line:prepend-zero}
        \return \opMed{$\txVotesMin[:\alpha]$}
    \EndFunction
    \Statex
    \vspace{-1mm}
    \Function{\opMaxPossibleTs{\txVotesMax}}{}\label{line:fn-compute-max-begin}
        \For{$\rep_j \in \allreps$}\label{line:assing-missing-max-begin}
            \State{\textbf{if} $\txVotesMax[\rep_j] = \bot$ \textbf{then} $\txVotesMax \gets \txVotesMax \concat [\infty]$} \label{line:fill-vote-inf}
        \EndFor\label{line:assing-missing-max-end}
        \State{\text{sort \txVotesMax in increasing order of timestamps}}
        \Let{\txVotesMax}{\txVotesMax \concat [ \infty, \stackrel{\beta \text{ times}}{\ldots}, \infty ]}\Comment{omitted altogether if $\beta=0$} \label{line:apend-infty}
        \return \opMed{$\txVotesMax[-\alpha:]$}
    \EndFunction\label{line:fn-compute-max-end}
    \Statex
    \vspace{-1mm}
    \Function{\opMinPossibleTsForNewTx{\mrt}}{}\label{line:fn-compute-pp-begin}
        \State{\text{sort \mrt in increasing order}}
        \Let{\mrt}{ [ 0, \stackrel{\beta \text{ times}}{\ldots}, 0 ] \concat \mrt}\Comment{omitted altogether if $\beta=0$}
        \return \opMed{$\mrt[:\alpha]$}
    \EndFunction\label{line:fn-compute-pp-end}
    \vspace{-1mm}
    \Statex
    \Function{\opMed{Y}}{}\label{}
        \return $Y[\lfloor \,|Y|/2 \, \rfloor]$
    \EndFunction
\end{algorithmic}
\end{algorithm}

\dotparagraph{Initialization}
The state of a client is shown in \Cref{alg:client-1} in lines~\ref{line:client-state-begin}--\ref{line:client-state-end}.
The state contains the identifiers and public keys of all replicas, \mrt, \nextsn, \tsps, \Pod, \ppcert, and \txcert.
Variable \tsps is a map from transactions \tx to a map from replicas \rep to timestamps \tsp.
The state gets initialized in lines~\ref{line:client-init-start}--\ref{line:client-init-end}.
At initialization the client also sends a \msgconnect{} message to each replica, which initiates a streaming connection from the replica to the client.

\dotparagraph{Receiving votes}
A client maintains a connection to each replica and receives votes through \msgrecord{(\tx, \tsp, \sn, \sig, \rep_j)} messages (lines~\ref{line:receive-vote-begin}--\ref{line:receive-vote-end}).
When a vote is received from replica $\rep_j$, the client first verifies the signature \sig under $\rep_j$'s public key (\cref{line:client-verify-sig-begin}). If invalid, the vote is ignored.
Then the client verifies that the vote contains the next sequence number it expects to receive from replica $\rep_j$ (\cref{line:client-check-sn}). If this is not the case, the vote is \emph{backlogged} and given again to the client at a later point (the backlogging functionality is not shown in the pseudocode).
The client then checks the vote against previous votes received from $\rep_j$.
First, \tsp must be greater or equal to $\mrt_j$, the most recent timestamp returned by replica $\rep_j$ (\cref{line:check-old-ts}).
Second, the replica must have not previously sent a different timestamp for \tx (\cref{line:check-dup-tsp}).
If both checks pass, the client updates $\mrt[j]$ (\cref{line:upd-mrt}) and $\tsps[\tx][\rep_j]$ (\cref{line:rec-tsp}) with \tsp.
The client also updates \ppcert and \alltxcert (\cref{line:upd-ppcert,line:client-store-vote-txcert}) for each valid vote.

If any of these checks fail, the client ignores the vote, since both of these cases constitute \emph{accountable} faults: In the first case, the client can use the message $\msgrecord{(\tx, \tsp, \sn, \sig, \rep_j)}$ and the vote it received when it updated $\mrt[\rep_j]$ to prove that $\rep_j$ has misbehaved. In the second case, it can use $\msgrecord{(\tx, \tsp, \sn, \sig, \rep_j)}$ and the previous vote it has received for \tx.
The \identify{} function we show in \Cref{alg:identify} can detect such misbehavior. However, in this paper we formalize accountability conditioned on safety being violated (\Cref{def:accountable-safety}), hence we do not further explore this.


\dotparagraph{Writing to and reading from \podtm}
Clients interact with a \podtm using the \opwrite{\tx} and \opreadall{} functions.
In order to write a transaction \tx, a client sends \msgwrite{\tx} to each replica (lines \ref{line:write-tx-begin}--\ref{line:write-tx-end}).
Since the construction is stateful and streaming, the client state contains at all times the latest view the client has of the \podtm. Hence, \opreadall{} operates on the local state (lines \ref{line:read-begin}--\ref{line:read-end}).
It returns all the transactions the client has received so far and their traces, and the current past-perfect round \rperf.
We will show the details of $\opComputeTxSet{}$ in \Cref{alg:client-statistics}.
As per the\podtm interface, \opreadall{} also returns auxiliary data \cert, which in the implementation of \podtm-core has two parts: the past-perfection certificate \ppcert and a list of transaction certificates \txcert (\cref{line:client-construct-cert}).
Note that $\tsps.\opKeys{}$ on \cref{line:loop-over-txs} returns all entries in \tsps.

\dotparagraph{Computing the trace values and the past-perfect round}
In \Cref{alg:client-statistics} we show function \opComputeTxSet{}, used to compute the current transaction set from the timestamps \tsps received so far.
A transaction becomes \confirmed when the client receives \alpha votes for \tx from different replicas (line~\ref{line:confirm-tx-begin}), in which case \rconf is the median of all received timestamps (line~\ref{line:assign-median}).
The computation of $\rmin, \rmax$, and \rperf is done using the functions \opMinPossibleTs{}, \opMaxPossibleTs{}, and \opMinPossibleTsForNewTx{}, respectively.

Function \opMinPossibleTs{} gets as input the timestamps \txVotesMin from each replica on \tx and the most recent timestamps \mrt from the replicas. It fills a missing timestamp from replica $\rep_j$ with $\mrt[\rep_j]$ (\cref{line:fill-vote-lts}),
the minimum timestamp that can ever be  accepted from $\rep_j$ (smaller values will not pass the check in \cref{line:check-old-ts} of \Cref{alg:client-1}).
It then prepends \beta times the $0$ value (\cref{line:prepend-zero}), pessimistically assuming that up to $\beta$ replicas will try to bias \tx by sending a timestamp $0$ to other clients, which only happens if replicas may be Byzantine, i.e., if $\beta > 0$.
It then returns the median of the \alpha smallest timestamps, which, again pessimistically, are the smallest timestamps another client may use to confirm \tx.

Function \opMaxPossibleTs{} is analogous, filling a missing vote with $\infty$ (\cref{line:fill-vote-inf})
and appending the $\infty$ value (\cref{line:apend-infty}), the worst-case timestamp that Byzantine replicas may send to other clients,
and returning the median of the \alpha largest timestamps.

Finally, \opMinPossibleTsForNewTx{} is similar to \opMinPossibleTs{} but it operates on the timestamps \mrt, instead of votes on a specific transaction. Hence, since an honest client will not accept a timestamp smaller than \mrt on any future transaction (\cref{line:check-old-ts} of \Cref{alg:client-1}), the returned value bounds from below the \confirmed round that \emph{any} honest client can ever assign to a transaction \emph{not yet seen}.


\subsection{Validation function}
The purpose of the validation function \opValid{} is to allow a client, which is not necessarily communicating with the \podtm replicas, to verify that a given pod data structure \Pod satisfies the security properties of \podtm (\Cref{def:podsec}).

\begin{algorithm}[h]
    \caption{Function \opValid{\Pod, \cert} for \podtm-core. Code for a \emph{verifier}, which can be a \podtm client not communicating with the \podtm replicas.}
    \label{alg:pod-core-valid}
    \begin{algorithmic}[1]
        \State{\textbf{State:} Same as in \Cref{alg:client-1}, includes $\{\rep_1, \ldots, \rep_n\}$, \tsps, \mrt.}
        \Statex
        \vspace{-1mm}
        \Function{\opValid{\Pod, \cert}}{}
            \Let{(\ppcert, \alltxcert) }{\cert} \label{line:valid-extract-cert} \Comment{$\ppcert: \rep \rightarrow \vote, \; \alltxcert: \tx \rightarrow \txcert, \; \txcert: \rep \rightarrow \vote$}
            \State{\opInitClient{}} \Comment{shown in \Cref{alg:client-1}}
            \Let{\allVotes}{\bigcup_{\tx \in \alltxcert} \left(\alltxcert[\tx].\opValues{}\right)}
            \For{$(\tx, \tsp, \sn, \sig, \rep_j) \in \allVotes$ \textbf{in increasing order of} $\sn$}
                \State{\textbf{require} \opProcessVote{$\tx, \tsp, \sn, \sig, \rep_j$}} \Comment{shown in \Cref{alg:client-1}, updates local state \tsps, \mrt}
            \EndFor
            \State{\textbf{require} $\Pod.\txSet = \opComputeTxSet{\tsps, \mrt}$} \label{line:valid-check-begin} \Comment{shown in \Cref{alg:client-statistics}}
            \State{\textbf{require} $\Pod.\rperf = \opMinPossibleTsForNewTx{\mrt}$} \label{line:valid-check-end} \Comment{shown in \Cref{alg:client-statistics}}
            \For{$(\tx, \tsp, \sn, \sig, \rep_j) \in \ppcert.\opValues{}$} \label{line:valid-check-ppcert-begin}
                \State{\textbf{require} $(\tx, \tsp, \sn, \sig, \rep_j) \in \allVotes$}
                \State{\textbf{require} $\sn = \op{max}_{\sn'}((\cdot, \cdot, \sn', \cdot, \rep_j) \in \allVotes)$}
            \EndFor\label{line:valid-check-ppcert-end}
        \EndFunction
    \end{algorithmic}
    \end{algorithm}

The function \opValid{} for \podtm-core is shown in \Cref{alg:pod-core-valid}.
The idea is to have the verifier repeat the logic of an honest client.
The verifier is initialized in the same way as in \Cref{alg:client-1} -- importantly, it knows the identifiers and public keys of \podtm replicas.
Function \opValid{} takes as input a pod data structure \Pod and auxiliary data \cert,
which is expected to contain two parts, a \emph{past-perfection certificate} \ppcert and a collection of \emph{transaction certificates} $\alltxcert$, once for each transaction in \Pod.\txSet (\cref{line:valid-extract-cert}).
Both contain vote messages, as constructed by a \podtm client in lines \ref{line:upd-ppcert} and \ref{line:client-store-vote-txcert} of \Cref{alg:client-1}.
The verifier processes each vote in order of increasing sequence number \sn using function \opProcessVote{}. If any vote is invalid, \opValid{} returns $\false$.
Observe that if the votes are valid the verifier will have updated its local \tsps and \mrt variables with the same values as the \podtm client that constructed \Pod.
Finally, the verifier computes the transaction set \txSet and the past-perfect round \rperf (using its local \tsps and \mrt variables) and requires that the values match the ones in \Pod (lines~\ref{line:valid-check-begin}--\ref{line:valid-check-end}).

Finally, the verifier also verifies the past-perfection certificate. Given that the previous checks have passed, we require that each vote in \ppcert is contained in one of the transaction certificates in \alltxcert and has the maximum sequence number received from the client that sent the vote (lines~\ref{line:valid-check-ppcert-begin}--\ref{line:valid-check-ppcert-end}).
As we have remarked earlier, \ppcert can be derived from \alltxcert by taking the union of certificates $\alltxcert$ for all transactions and keeping the most recent vote for each replica, in which case the checks on lines~\ref{line:valid-check-ppcert-begin}--\ref{line:valid-check-ppcert-end} can be omitted. We maintain the past-perfection certificate for readability and simplicity in the proofs.

\subsection{Analysis}

\begin{theorem}[\podtm-core security]\label{thm:pod-security}
    Assume that the network is partially synchronous with actual network delay $\delta$, that \beta is the number of Byzantine replicas, \gamma the number of omission-faulty replicas, $\alpha = n - \beta - \gamma$ the confirmation threshold, and $n \geq 5 \beta + 3 \gamma + 1$ the total number of replicas. Protocol \podtm-core (\Cref{const:pod-core}), instantiated with a EUF-CMA secure signature scheme, the \opValid{} function shown in \Cref{alg:pod-core-valid}, and the \identify{} function described in \Cref{alg:identify}, is a responsive secure \podtm (\Cref{def:podsec}) with Confirmation within $\parConf = 2 \delta$, Past-perfection within $\parPerf = \delta$ and  \beta-accountable safety (\Cref{def:accountable-safety}), \ewnp.
\end{theorem}
\begin{proof}
    Shown in Appendix~\ref{app:pod-security}.
\end{proof}


\section{Evaluation} \label{sec:evaluation}
To validate our theoretical results regarding optimal latency in Protocol \podtm-core, we implement\footnote{Our prototype implementation is available at \url{https://github.com/commonprefix/pod-experiments}} a prototype \podtm-core in Rust 1.85.
Our benchmarks measure the
end-to-end confirmation latency of a transaction from the moment it is written by client until it is read as
confirmed by another client in a different continent, both interacting with replicas distributed around the world.
Specifically, the latency is computed as the difference between the timestamp
recorded by the reading client upon receiving sufficiently many votes (quorum size $\alpha$) from different replicas and the initial timestamp recorded by the writing client.
We present the results in \Cref{fig:latency}.

\begin{figure}[h!]
    \centering
    \includegraphics[width=0.9\textwidth]{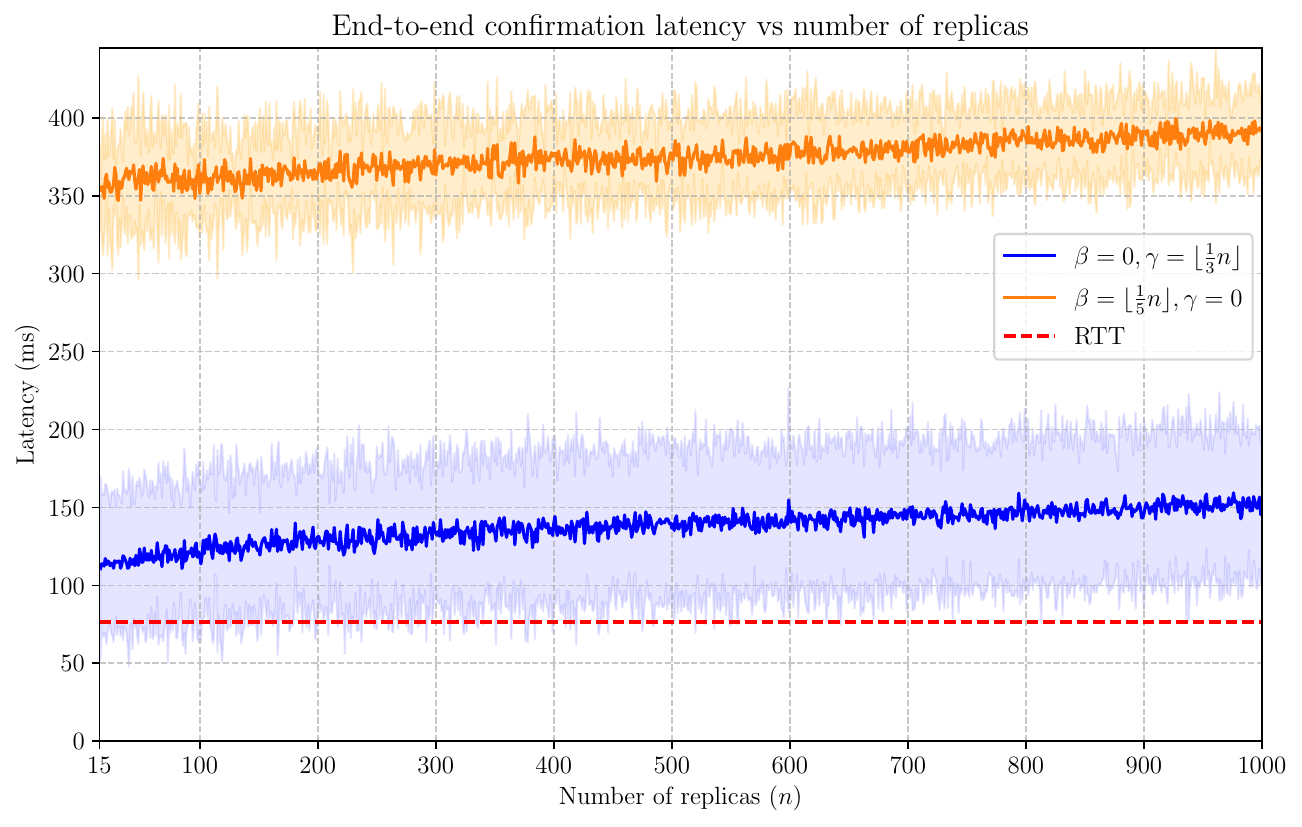}
    \vspace{-5mm}
    \caption{End-to-end confirmation latency from a writing client to a reading client as a transaction traverses across $n = 15, \ldots, 1000$ replicas,
    for two reading clients:
    (1) a client that expects up to $\gamma = \lfloor\frac{1}{3}n\rfloor$ omission faults (blue line, below), and (2) a client that expects up to $\beta = \lfloor\frac{1}{5}n\rfloor$ Byzantine faults (orange line, above).
    We also plot the physical network round-trip time (RTT) between the reading client and the writing client, which is 76ms (dashed red line).
    A 95\% confidence interval is shown for each experiment (shaded area).}
    \label{fig:latency}
\end{figure}

The implementation follows a client-server architecture where each replica maintains two TCP listening sockets:
one for the reading client connection and one for the writing client connection. Upon receiving a transaction
payload from a writer, the replica creates a tuple containing the payload, a sequence number, and the current local timestamp.
The replica then signs this tuple using a Schnorr signature\footnote{\url{https://crates.io/crates/secp256k1}} on secp256k1 curve, appends it to its local log, and forwards
the signed tuple to the reading client.
Replicas are deployed round-robin across seven
AWS regions: eu-central-1 (Frankfurt), eu-west-2 (London), us-east-1 (N. Virginia), us-west-1 (N. California),
ca-central-1 (Canada), ap-south-1 (Mumbai), and ap-northeast-2 (Seoul). Each replica is deployed on a
t2.medium EC2 instance (2 vCPUs, 4GB RAM) and is initialized with user data that contains
the replica's unique secret signing key.

We implement two types of clients. The writing client establishes connections to all replicas, records
the timestamp (in its local view) right before sending the transaction and sends transaction
payloads to each replica. The reading client maintains connections to all replicas, validates incoming signed
transactions, and records the timestamp (in its local view) upon receiving a quorum of valid signatures for a particular transaction.
We deploy the reading
client in eu-west-2 (London) and the writing client in us-east-1 (N. Virginia), both initialized with the complete
list of replica information (IP addresses, public keys).

We conduct experiments with two different values for the quorum size $\alpha = 1 - \beta - \gamma$: (1) $\beta =0 $ and $\gamma = \lfloor \frac{1}{3} n \rfloor$, for a client that only expects omission faults,
and (2) $\beta = \lfloor\frac{1}{5}n\rfloor$ and $\gamma = 0$, for a client that expects Byzantine faults.
We repeat the experiments for
different numbers of replicas ($n = 15, \ldots, 1000$).
We repeat each experiment five times and report the mean latency and a 95\% confidence interval.


As shown in Figure~\ref{fig:latency}, our experimental results demonstrate that the latency remains largely independent
of the number of replicas.
The reading client reports a transaction as confirmed as soon as the fastest
$\alpha$ replicas have responded, which gives rise to the happy artifact that the 1 - $\alpha$ slowest replicas
do not slow down confirmation.
This also explains why the omission-fault experiment exhibits lower latency than the Byzantine experiment.
Even with 1000 replicas the mean confirmation latency is 138ms for the omission-fault experiment and 375ms for the Byzantine experiment.
This approximates
the physical network round-trip time between the reading client and the writing client that stands at 76ms.

\section{Auctions on \podtm through the \bidsettm protocol}\label{sec:auctions}
In this section, we show how single-shot distributed auctions can be implemented on top of \podtm.
This is achieved through \bidsettm, a primitive for collecting a set of bids.
The idea is as follows. A pre-appointed \emph{sequencer} runs the auction, but the bids are collected from \podtm using a \bidsettm protocol. The past-perfection property of \podtm renders the sequencer unable to censor bids: when it creates an output, all timely and honestly-written bids \emph{must} be in it, otherwise the sequencer has provably misbehaved and can be held accountable.
We first define \bidsettm and then construct it using an underlying \podtm.

\begin{remark}[Implicit sub-session identifiers]
	We assume that each instance of the \bidsettm-core protocol is identified by a unique sub-session identifier (ssid). All messages written to the underlying \podtm are concatenated with the ssid.
\end{remark}

\begin{definition}[\bidsettm protocol]\label{def:auction-framework}
    A \bidsettm \emph{protocol} has a \emph{starting time} parameter \tzero and exposes the following interfaces to \emph{bidder} and \emph{consumer} parties:
    \begin{itemize}
        \item function $\opSubmitBid{\bid}$: It is called by a bidder at round \tzero to submit a bid \bid.
        \item event $\opResult{\bag, \bidcert}$: It is an event generated by a consumer. It contains a \emph{bid-set} \bag, which is a set of bids, and auxiliary information \bidcert.
    \end{itemize}
\end{definition}

\noindent
A \bidsettm protocol satisfies the following liveness and safety properties:
\begin{description}[leftmargin=4mm, labelindent=4mm]
    \item[(Liveness) Termination within \parAuctionTerm:]
    An honest consumer generates an event $\opResult{\bag, \bidcert}$ by round $\tzero + \parAuctionTerm$.

    \item [(Safety) Censorship resistance:]
    If an honest bidder calls \opSubmitBid{\bid}
    and an honest consumer generates an event $\opResult{\bag, $\cdot$}$,
    then $\bid \in \bag$.

    \item[(Safety) Weak consistency:]
    If two honest consumers generate \opResult{$\bag_1, \cdot$} and \opResult{$\bag_2, \cdot$} events, such that $\bag_1 \neq \emptyset$ and $\bag_2 \neq \emptyset$, then $\bag_1 = \bag_2$.

\end{description}



\begin{protocol}[\bidsettm-core]\label{const:bidset}
	Protocol \bidsettm-core is parameterized by an integer $\Delta$ (looking ahead, we will prove security in synchrony, \textit{i.e.}, assuming the network delay $\delta$ is smaller than $\Delta$) and assumes digital signatures and a \podtm with $\delta$-timeliness, $\parPerf=\delta$ and $\parConf=2\delta$. At time \tzero, all parties start executing Algorithms~\ref{alg:bid-set-bidder}--\ref{alg:bid-set-consumer}.
	A pre-appointed \emph{sequencer} is responsible to reading the \podtm and writing back to it when a specific condition is met. For example, when instantiating \bidsettm-core on top of \podtm-core, a replica can act as sequencer.
\end{protocol}

\vspace{-5mm}
\begin{algorithm}[h!]
    \caption{\bidsettm-core: Code for a bidder. It runs a client for a \podtm-core instance \podInst.}
    \label{alg:bid-set-bidder}
    \begin{algorithmic}[1]
        \Function{\opSubmitBid{\bid}}{}
            \State{\podInst.\opwrite{\bid}}
        \EndFunction
\end{algorithmic}
\end{algorithm}

\vspace{-5mm}
\begin{algorithm}[h!]
    \caption{\bidsettm-core: Code for the sequencer. It runs a client for a \podtm-core instance \podInst, and $\sk_a$ denotes the secret key of the sequencer.}
    \label{alg:bid-set-sequencer}
    \begin{algorithmic}[1]
        \Function{\opRunAuction{}}{}
            \Let{\left( (\txSet, \rperf), (\ppcert, \alltxcert) \right)}{\podInst.\opreadall{}}
            \While{$\rperf \leq \tzero + \Delta$}\label{line:bidset-wait-pp}
                \Let{\left( (\txSet, \rperf), (\ppcert, \alltxcert) \right)}{\podInst.\opreadall{}}
            \EndWhile{}
            \State{$\bag \gets \{\tx \mid (\tx, \cdot, \cdot , \cdot) \in \txSet\}  ;\; \bidcert \gets \ppcert$}\label{line:bidset-create-bag}
            \Let{\sigma}{\opSign{\sk_a, (\bag, \bidcert)}} \label{line:bidset-sequencer-sign}
            \Let{\tx}{\msgbidresult{(\bag, \bidcert, \sigma)}}\label{line:bidset-write-bag}
            \State{\podInst.\opwrite{\tx}} 
        \EndFunction
\end{algorithmic}
\end{algorithm}

\vspace{-5mm}
\begin{algorithm}[h!]
    \caption{\bidsettm-core: Code for a consumer. It runs a client for a \podtm-core instance \podInst.}
    \label{alg:bid-set-consumer}
    \begin{algorithmic}[1]
        \Function{\opReadResult{}}{}
            \Loop{}\label{line:bidset-wait-pp-consumer}
                \Let{((\txSet, \rperf), (\ppcert, \alltxcert))}{\podInst.\opreadall{}}
                    \If{$\exists (\tx, \cdot, \cdot, \rconf, \cdot) \in \txSet$ \textbf{:} $\tx = \msgbidresult{(\bag, \bidcert, \sigma)}$ \andl $\rconf \leq \tzero + 3 \Delta$} \label{line:bidset-find-confirmed-tx}
                    \State{\textbf{output event} \opResult{\bag, \bidcert}}
                \ElsIf{$\rperf > \tzero + 3 \Delta$}\label{line:bidset-read-result-pp}
                    \State{\textbf{output event} \opResult{$\emptyset$, \ppcert}}
                \EndIf
            \EndLoop\label{line:bidset-consumer-loop-end}
        \EndFunction
    \end{algorithmic}
    \end{algorithm}

\vspace{-5mm}

A bidder (\Cref{alg:bid-set-bidder}) submits a bid by writing it on the \podtm  at round \tzero.
The sequencer (\Cref{alg:bid-set-sequencer}) waits until the \podtm returns a past-perfect round larger than $\tzero + \Delta$ (\cref{line:bidset-wait-pp}) and then constructs the bid-set \bag from the set of transactions in \txSet (\cref{line:bidset-create-bag}).
The sequencer concludes by signing \bag and \bidcert (which can be used as evidence, in case of a safety violation) and writing  \msgbidresult{(\bag, \bidcert, \sigma)} on \podInst.

The code for a consumer is shown in \Cref{alg:bid-set-consumer}. The consumer waits until one of the following two conditions is met.
First, a \emph{confirmed} transaction \msgbidresult{(\bag, \bidcert, \sigma)} appears in \txSet , for which $\rconf \leq \tzero + 3 \Delta$ (\cref{line:bidset-find-confirmed-tx}), in which case it outputs bid-set \bag as result.
Second, a round higher than $\tzero + 3 \Delta$ becomes past-perfect in \podtm (\cref{line:bidset-read-result-pp}) without a confirmed \msgbidresult{} transaction appearing, in which case it outputs $\bag=\emptyset$.

As an intuition on how \bidsettm-core achieves censorship resistance, we observe the following.
The \emph{$\delta$-timeliness property} of \podtm (\Cref{def:pod-timely}), given that $\delta \leq \Delta$, ensures that bids of honest parties will have a \confirmed round $\rconf \leq \tzero + \Delta$.
Now, the sequencer may only produce a valid bid-set when \podtm returns a past-perfect round larger than $\tzero + \Delta$ (\cref{line:bidset-wait-pp}), and the output \emph{must} contain a certificate \ppcert that proves this. However, if the certificate is valid, then the sequencer must have \emph{provably} seen the bids of honest parties in \txSet (we remind that the votes of replicas on \podtm-core are chained using sequence numbers), and thus \bag must contain all bids with $\rconf \leq \tzero + \Delta$. If any party presents a transaction certificate \txcert for some transaction $\tx^*$ with $\rconf^* \leq \tzero + \Delta$, but $\tx^* \not\in \bag$, then the sequencer can be held accountable. We show the detailed proof in \Cref{lem:bidset-censorship} and \Cref{lem:bidset-accountability}.

Regarding liveness, \cref{line:bidset-wait-pp} of \Cref{alg:bid-set-sequencer} becomes true in the view of sequencer by round $\tzero + \Delta + \delta$ (from the \emph{past-perfection within $\parPerf=\delta$} property of \podtm-core),
hence \Cref{alg:bid-set-sequencer} for an honest sequencer terminates by that round.
Observe also that the transaction \msgbidresult{(\bag, \bidcert, \sigma)} becomes confirmed in the view of all honest clients by round $\tzero + \Delta + 3\delta$ (from the \emph{confirmation within $\parConf=2\delta$} property), and it will have a confirmed round $\rconf \leq \tzero + \Delta + 2\delta$ (from the \emph{$\delta$-timeliness} property).
Hence, if the network is synchronous and the sequencer honest, the condition in \cref{line:bidset-find-confirmed-tx} of \Cref{alg:bid-set-consumer} becomes true at round at most $\tzero + \Delta + 3\delta$,
Even if the sequencer is malicious, from the \emph{past-perfection within $\parPerf=\delta$} property of \podtm, the condition in \cref{line:bidset-read-result-pp} will become true at latest at round $\tzero + 3 \Delta + \delta$, hence \bidsettm-core achieves \emph{termination within $\parAuctionTerm=3\Delta + \delta$}.


\begin{theorem}[Bidset security]\label{thm:bidset-construction-secure}
Assuming a synchronous network where $\delta \leq \Delta$,
protocol \bidsettm-core (Construction~\ref{const:bidset}) instantiated with a digital signature and a  secure \podtm protocol that satisfies the \emph{past-perfection within $\parPerf = \delta$}, \emph{confirmation within $\parConf=2\delta$} and \emph{$\delta$-timeliness} properties, is a secure \bidsettm protocol satisfying \emph{termination within $\parAuctionTerm = 3\Delta + \delta$}.
It satisfies accountable safety with an \identifySequencer{} function that identifies a malicious sequencer.
\end{theorem}
\vspace{-3mm}
\begin{proof} The proof and \identifySequencer{} are shown in Appendix~\ref{app:bidsetsec}.
\end{proof}

\begin{remark}
    Observe that \bidsettm-core terminates within $\parAuctionTerm = 3\Delta + \delta$ in the worst case, but, if the sequencer is honest, then it terminates within $\parAuctionTerm = \Delta + 3\delta$.
    Moreover, \bidsettm-core is not responsive because \Cref{alg:bid-set-sequencer} waits for a fixed $\Delta$ interval.
    This step can be optimized if the set of bidders is known (i.e., by requiring them to pre-register), which allows for the protocol to be made optimistically responsive (\textit{i.e.}, $\parAuctionTerm = 4\delta$) when all bidders and the sequencer are honest.
\end{remark}
\vspace{-2mm}
\dotparagraph{\textbf{Auctions using \bidsettm}} Building on a \bidsettm protocol, it is trivial to construct single-shot first price and second price open auctions as follows: 1. Bidders place their open bids $\bid$ by calling \opSubmitBid{\bid}; 2. Consumers determine the winner by calling \opReadResult{} to obtain $\bag$ and outputting either the first or second highest bid. We conjecture that single-shot sealed bid auction protocols such as those of~~\cite{trustee,seal,ACNS:DavGenPou22,CCS:TAFWBM23,chitra_et_al:LIPIcs.AFT.2024.19,cryptoeprint:2024/1011} can also be instantiated on top of a \bidsettm protocol. Intuitively, this holds because such protocols first agree on a set of sealed bids and then execute extra steps to determine the winner. However, a formal analysis of sealed-bid auction protocols based on \bidsettm is left as future work.
\section{Discussion}
In this work we present \podtm, a novel consensus layer that finalizes transactions with the optimal one-round-trip latency by eliminating communication among replicas. Instead, clients read the system state by performing lightweight computation on logs retrieved from the replicas. As no replica has a particular role in \podtm (as compared to leaders, block proposers or miners in similar protocols), \podtm achieves censorship resistance by default, without any extra mechanisms or additional cost. Furthermore, replica misbehavior, such as voting in incompatible ways or censoring confirmed transactions, is accountable.

Regarding applications, we have presented an efficient and censorship-resistant auction mechanism, which leverages \podtm as a bulletin board. We show how the accountability, offered by \podtm, is also inherited by applications built on it -- the auctioneer cannot censor confirmed bids without being detected. Similar to auctions, \podtm can enable censorship-resistant voting applications -- \podtm guarantees that no single party or authority can censor or delay a valid vote.

Moreover, payments can be realized on top of \podtm. We leave the complete specification as future work, but outline here two ways in which this can be achieved. The first is by making the replicas stateful, in which case \podtm can directly support a protocol similar to FastPay~\cite{DBLP:conf/aft/BaudetDS20}. The second option is to implement the payment logic on the client side, hence leaving \podtm stateless. This can be achieved using the past-perfection property of \podtm: the sender of a payment  writes the payment transaction to \podtm; the recipient waits until the transaction becomes confirmed and its confirmed round becomes past-perfect; the recipient can then verify whether the sender has created a conflicting transaction before it. Compared to the solution of FastPay, the second approach has the advantage that clients do not need to maintain sequence numbers.

We remark that \podtm differs from standard notions of consensus because it does not offer an agreement property, neither to replicas nor to clients. A client reading the \podtm obtains a past-perfect round \rperf, and it is guaranteed to have received all transactions that obtained a confirmed round \rconf such that $\rconf \leq \rperf$. It is also guaranteed to have received all transactions that can potentially obtain an $\rconf \leq \rperf$ in the future, even though the transaction presently appears to the client as unconfirmed. However, the client cannot tell which unconfirmed transactions will become confirmed. Moreover, a transaction might appear confirmed to one client and unconfirmed to another (in this case, this will be transaction written by a malicious client).

\iflncs
\thispagestyle{plain}
\fi

\ifccs
  \bibliographystyle{ACM-Reference-Format}
\elseiflncs
  \bibliographystyle{splncs04}
\else
  \bibliographystyle{abbrv}
\fi


\pdfbookmark[section]{References}{references}
\bibliography{abbrev3,crypto,biblio}

@string{ieee =                  {IEEE}}

@string{springer =              "Springer"}

@string{dagstuhl =              "Schloss Dagstuhl"}

@string{acm =                   "{ACM}"}

@article{GMR88,
	author       = {Shafi Goldwasser and
	Silvio Micali and
	Ronald L. Rivest},
	title        = {A Digital Signature Scheme Secure Against Adaptive Chosen-Message
	Attacks},
	journal      = {{SIAM} J. Comput.},
	volume       = {17},
	number       = {2},
	pages        = {281--308},
	year         = {1988},
	url          = {https://doi.org/10.1137/0217017},
	doi          = {10.1137/0217017},
	bibsource    = {dblp computer science bibliography, https://dblp.org}
}

@article{DBLP:journals/dc/GuerraouiKMPS22,
	author       = {Rachid Guerraoui and
	Petr Kuznetsov and
	Matteo Monti and
	Matej Pavlovic and
	Dragos{-}Adrian Seredinschi},
	title        = {The consensus number of a cryptocurrency},
	journal      = {Distributed Comput.},
	volume       = {35},
	number       = {1},
	pages        = {1--15},
	year         = {2022},
	url          = {https://doi.org/10.1007/s00446-021-00399-2},
	doi          = {10.1007/S00446-021-00399-2},
	bibsource    = {dblp computer science bibliography, https://dblp.org}
}

@inproceedings{DBLP:conf/dsn/CollinsGKKMPPST20,
	author       = {Daniel Collins and
	Rachid Guerraoui and
	Jovan Komatovic and
	Petr Kuznetsov and
	Matteo Monti and
	Matej Pavlovic and
	Yvonne{-}Anne Pignolet and
	Dragos{-}Adrian Seredinschi and
	Andrei Tonkikh and
	Athanasios Xygkis},
	title        = {Online Payments by Merely Broadcasting Messages},
	booktitle    = {50th Annual {IEEE/IFIP} International Conference on Dependable Systems
	and Networks, {DSN} 2020, Valencia, Spain, June 29 - July 2, 2020},
	pages        = {26--38},
	publisher    = {{IEEE}},
	year         = {2020},
	url          = {https://doi.org/10.1109/DSN48063.2020.00023},
	doi          = {10.1109/DSN48063.2020.00023},
	bibsource    = {dblp computer science bibliography, https://dblp.org}
}

@inproceedings{DBLP:conf/aft/BaudetDS20,
	author       = {Mathieu Baudet and
	George Danezis and
	Alberto Sonnino},
	title        = {FastPay: High-Performance Byzantine Fault Tolerant Settlement},
	booktitle    = {{AFT}},
	pages        = {163--177},
	publisher    = {{ACM}},
	year         = {2020}
}

@article{DBLP:journals/corr/abs-1909-10926,
	author       = {Jakub Sliwinski and
	Roger Wattenhofer},
	title        = {{ABC:} Asynchronous Blockchain without Consensus},
	journal      = {CoRR},
	volume       = {abs/1909.10926},
	year         = {2019},
	url          = {http://arxiv.org/abs/1909.10926},
	eprinttype    = {arXiv},
	eprint       = {1909.10926},
	bibsource    = {dblp computer science bibliography, https://dblp.org}
}

@InProceedings{chitra_et_al:LIPIcs.AFT.2024.19,
	author =	{Chitra, Tarun and Ferreira, Matheus V. X. and Kulkarni, Kshitij},
	title =	{{Credible, Optimal Auctions via Public Broadcast}},
	booktitle =	{6th Conference on Advances in Financial Technologies (AFT 2024)},
	pages =	{19:1--19:16},
	series =	{Leibniz International Proceedings in Informatics (LIPIcs)},
	ISBN =	{978-3-95977-345-4},
	ISSN =	{1868-8969},
	year =	{2024},
	volume =	{316},
	editor =	{B\"{o}hme, Rainer and Kiffer, Lucianna},
	publisher =	{Schloss Dagstuhl -- Leibniz-Zentrum f{\"u}r Informatik},
	address =	{Dagstuhl, Germany},
	URL =		{https://drops.dagstuhl.de/entities/document/10.4230/LIPIcs.AFT.2024.19},
	URN =		{urn:nbn:de:0030-drops-209550},
	doi =		{10.4230/LIPIcs.AFT.2024.19}
}

@ARTICLE{seal,
	author={Bag, Samiran and Hao, Feng and Shahandashti, Siamak F. and Ray, Indranil Ghosh},
	journal={IEEE Transactions on Information Forensics and Security},
	title={SEAL: Sealed-Bid Auction Without Auctioneers},
	year={2020},
	volume={15},
	number={},
	pages={2042-2052},
	doi={10.1109/TIFS.2019.2955793}}

@misc{cryptoeprint:2024/1011,
	author = {Chaya Ganesh and Shreyas Gupta and Bhavana Kanukurthi and Girisha Shankar},
	title = {Secure Vickrey Auctions with Rational Parties},
	howpublished = {Cryptology {ePrint} Archive, Paper 2024/1011},
	year = {2024},
	doi = {10.1145/3658644.3670311},
	url = {https://eprint.iacr.org/2024/1011},
	note = {To appear at {CCS} 2024}
}

@InProceedings{trustee,
	author="Galal, Hisham S.
	and Youssef, Amr M.",
	editor="Bracciali, Andrea
	and Clark, Jeremy
	and Pintore, Federico
	and R{\o}nne, Peter B.
	and Sala, Massimiliano",
	title="Trustee: Full Privacy Preserving Vickrey Auction on Top of Ethereum",
	booktitle="Financial Cryptography and Data Security",
	year="2020",
	publisher="Springer International Publishing",
	address="Cham",
	pages="190--207",
	isbn="978-3-030-43725-1"
}

@inproceedings{DBLP:conf/focs/GarayKKO07,
  author       = {Juan A. Garay and
                  Jonathan Katz and
                  Chiu{-}Yuen Koo and
                  Rafail Ostrovsky},
  title        = {Round Complexity of Authenticated Broadcast with a Dishonest Majority},
  booktitle    = {{FOCS}},
  pages        = {658--668},
  publisher    = {{IEEE} Computer Society},
  year         = {2007}
}

@inproceedings{DBLP:conf/podc/YinMRGA19,
  author       = {Maofan Yin and
                  Dahlia Malkhi and
                  Michael K. Reiter and
                  Guy Golan{-}Gueta and
                  Ittai Abraham},
  title        = {HotStuff: {BFT} Consensus with Linearity and Responsiveness},
  booktitle    = {{PODC}},
  pages        = {347--356},
  publisher    = {{ACM}},
  year         = {2019}
}

@inproceedings{DBLP:conf/eurosys/DanezisKSS22,
  author       = {George Danezis and
                  Lefteris Kokoris{-}Kogias and
                  Alberto Sonnino and
                  Alexander Spiegelman},
  title        = {Narwhal and Tusk: a DAG-based mempool and efficient {BFT} consensus},
  booktitle    = {EuroSys},
  pages        = {34--50},
  publisher    = {{ACM}},
  year         = {2022}
}

@article{DBLP:journals/iacr/MalkhiN23,
  author       = {Dahlia Malkhi and
                  Kartik Nayak},
  title        = {Extended Abstract: HotStuff-2: Optimal Two-Phase Responsive {BFT}},
  journal      = {{IACR} Cryptol. ePrint Arch.},
  pages        = {397},
  year         = {2023}
}

@inproceedings{DBLP:conf/fc/GelashviliKSSX22,
  author       = {Rati Gelashvili and
                  Lefteris Kokoris{-}Kogias and
                  Alberto Sonnino and
                  Alexander Spiegelman and
                  Zhuolun Xiang},
  title        = {Jolteon and Ditto: Network-Adaptive Efficient Consensus with Asynchronous
                  Fallback},
  booktitle    = {Financial Cryptography},
  series       = {Lecture Notes in Computer Science},
  volume       = {13411},
  pages        = {296--315},
  publisher    = {Springer},
  year         = {2022}
}

@article{DBLP:journals/corr/abs-2401-01791,
  author       = {Isaac Doidge and
                  Raghavendra Ramesh and
                  Nibesh Shrestha and
                  Joshua Tobkin},
  title        = {Moonshot: Optimizing Chain-Based Rotating Leader {BFT} via Optimistic
                  Proposals},
  journal      = {CoRR},
  volume       = {abs/2401.01791},
  year         = {2024}
}

@article{DBLP:journals/corr/abs-2209-05633,
  author       = {Alexander Spiegelman and
                  Neil Giridharan and
                  Alberto Sonnino and
                  Lefteris Kokoris{-}Kogias},
  title        = {Bullshark: The Partially Synchronous Version},
  journal      = {CoRR},
  volume       = {abs/2209.05633},
  year         = {2022}
}

@article{DBLP:journals/corr/abs-2310-14821,
  author       = {Kushal Babel and
                  Andrey Chursin and
                  George Danezis and
                  Lefteris Kokoris{-}Kogias and
                  Alberto Sonnino},
  title        = {Mysticeti: Low-Latency {DAG} Consensus with Fast Commit Path},
  journal      = {CoRR},
  volume       = {abs/2310.14821},
  year         = {2023}
}

@article{DBLP:journals/ipl/AguileraT99,
	author       = {Marcos Kawazoe Aguilera and
	Sam Toueg},
	title        = {A Simple Bivalency Proof that \emph{t}-Resilient Consensus Requires
	\emph{t} + 1 Rounds},
	journal      = {Inf. Process. Lett.},
	volume       = {71},
	number       = {3-4},
	pages        = {155--158},
	year         = {1999},
	url          = {https://doi.org/10.1016/S0020-0190(99)00100-3},
	doi          = {10.1016/S0020-0190(99)00100-3},
	bibsource    = {dblp computer science bibliography, https://dblp.org}
}

@inproceedings{DBLP:conf/crypto/GaziRR23,
	author       = {Peter Gazi and
	Ling Ren and
	Alexander Russell},
	editor       = {Helena Handschuh and
	Anna Lysyanskaya},
	title        = {Practical Settlement Bounds for Longest-Chain Consensus},
	booktitle    = {Advances in Cryptology - {CRYPTO} 2023 - 43rd Annual International
	Cryptology Conference, {CRYPTO} 2023, Santa Barbara, CA, USA, August
	20-24, 2023, Proceedings, Part {I}},
	series       = {Lecture Notes in Computer Science},
	volume       = {14081},
	pages        = {107--138},
	publisher    = {Springer},
	year         = {2023},
	url          = {https://doi.org/10.1007/978-3-031-38557-5\_4},
	doi          = {10.1007/978-3-031-38557-5\_4},
	bibsource    = {dblp computer science bibliography, https://dblp.org}
}

@InProceedings{bootclocksync,
	author="Widder, Josef",
	editor="Fich, Faith Ellen",
	title="Booting Clock Synchronization in Partially Synchronous Systems",
	booktitle="Distributed Computing",
	year="2003",
	publisher="Springer Berlin Heidelberg",
	address="Berlin, Heidelberg",
	pages="121--135",
	isbn="978-3-540-39989-6"
}

@InProceedings{clockoverview,
	author="Simons, Barbara",
	editor="Simons, Barbara
	and Spector, Alfred",
	title="An overview of clock synchronization",
	booktitle="Fault-Tolerant Distributed Computing",
	year="1990",
	publisher="Springer New York",
	address="New York, NY",
	pages="84--96",
	isbn="978-0-387-34812-4"
}

@article{psync,
	author = {Dwork, Cynthia and Lynch, Nancy and Stockmeyer, Larry},
	title = {Consensus in the presence of partial synchrony},
	year = {1988},
	issue_date = {April 1988},
	publisher = {Association for Computing Machinery},
	address = {New York, NY, USA},
	volume = {35},
	number = {2},
	issn = {0004-5411},
	url = {https://doi.org/10.1145/42282.42283},
	doi = {10.1145/42282.42283},
	journal = {J. ACM},
	month = apr,
	pages = {288–323},
	numpages = {36}
}

\appendix
\section{Definition of \podtm monotonicity}\label{app:monotonicity}

The property of \podtm monotonicity requires that, as time advances, $\rmin$ does not decrease, $\rmax$ does not increase and confirmed transactions remain confirmed.

\begin{definition}[\podtm monotonicity]
	A \podtm protocol satisfies \emph{pod monotonicity}, and is called a \emph{monotone \podtm}, if it is a secure \podtm, as per Definition~\ref{def:podsec}, and the following properties hold for any rounds $\round_1,\round_2 > \round_1$ and for any honest client~\client:
	\begin{description}
		\item [Past-perfection monotonicity:] It holds that $\podView{\client}{\round_2}.\rperf \geq \podView{\client}{\round_1}.\rperf$.
		\item [Transaction monotonicity:] If  transaction $\tx$ appears in $\podView{\client}{\round_1}.\txSet$, then $\tx$ appears in $\podView{\client}{\round_2}.\txSet$.
		\item [Confirmation-bounds monotonicity:] For every $\tx$ that appears in $\podView{\client}{\round_1}.\txSet$ with $\rmin,\rmax,\rconf$ and appears in  $\podView{\client}{\round_2}.\txSet$ with $\rmin',\rmax',\rconf'$,  it holds that $\rmin'\geq \rmin$, $\rmax' \leq \rmax$.

	\end{description}
\end{definition}

We now observe that a monotone \podtm protocol can be obtained from any secure \podtm protocol with stateful clients, and that monotonicity implies certain specific properties that may be useful for applications. In particular, our \podtm-core protocol naturally satisfies this property.

\begin{remark}
	Every secure \podtm can be transformed into a monotone \podtm if parties are stateful. Let $\round_1$ be the last round when an honest client $\client$ read the \podtm obtaining view $\podView{\client}{\round_1}$, which is stored as state until $\client$ reads the \podtm again. At any round $\round_2>\round_1$, if $\client$ reads the \podtm and obtains $\podView{\client}{\round_2}$, $\client$ can define a view $\overline{\podView{\client}{\round_2}}$ satisfying the properties of \podtm monotonicity:

	\begin{enumerate}
		\item If $\tx$ appears in $\podView{\client}{\round_1}$ with $\tx.\rmin,\tx.\rmax,\tx.\rconf,\tx.\txcert$, then $\tx$ appears in $\overline{\podView{\client}{\round_2}}$ with $\tx.\overline{\rmin}=\rmin,\tx.\overline{\rmax}=\rmax,\tx.\overline{\rconf}=\rconf,\tx.\overline{\txcert}=\txcert$.
		\item If $\tx$ appears in $\podView{\client}{\round_2}$ with $\tx.\rmin',\tx.\rmax',\tx.\rconf',\tx.\txcert'$ and does not appear in $\podView{\client}{\round_1}$, then $\tx$ appears in $\overline{\podView{\client}{\round_2}}$ with $\tx.\overline{\rmin}=\rmin',\tx.\overline{\rmax}=\rmax',\tx.\overline{\rconf}=\rconf',\tx.\overline{\txcert}=\txcert'$.
		\item For every $\tx$ that appears in $\podView{\client}{\round_1}.\txSet$ and in  $\podView{\client}{\round_2}.\txSet$ such that $\rmin'\geq \rmin$, $\rmax' \leq \rmax$, $\rconf' \geq \rconf$, update $\tx.\overline{\rmin}=\rmin',\tx.\overline{\rmax}=\rmax',\tx.\overline{\rconf}=\rconf',\tx.\overline{\txcert}=\txcert'$.
		\item If  $\podView{\client}{\round_2}.\rperf > \podView{\client}{\round_1}.\rperf$, then $\overline{\podView{\client}{\round_2}}.\rperf=\podView{\client}{\round_2}.\rperf $. Otherwise, $\overline{\podView{\client}{\round_2}}.\rperf=\podView{\client}{\round_1}.\rperf $.

	\end{enumerate}
\end{remark}


In the remarks below, we observe that \podtm monotonicity implies a number of useful properties about the monotonicity of past perfection and the values $\rmin,\rmax,\rconf$ associated to a transaction in the \podtm.

\begin{remark}[Confirmation monotonicity]\label{rem:confstick}
	Properties 2 and 3 of \podtm monotonicity imply that  for any honest client $\client$ and rounds $\round_1, \round_2 > \round_1$,
	if $\tx \in \podView{\client}{\round_1}$ and $\tx.\rconf \neq \bot$, then $\tx \in \podView{\client}{\round_2}$ and $\tx.\rconf \neq \bot$.
\end{remark}



\begin{remark}
	Observe that the \emph{confirmation monotonicity} property in Remark~\ref{rem:confstick} is a specific version of a more general \emph{common subset} property, which would demand the condition for any two honest clients $\client_1, \client_2$.
\end{remark}
\section{Security of Protocol \podtm-core under a Continuum of Byzantine and Omission faults}\label{app:pod-security}
In order to prove Theorem~\ref{thm:pod-security} and establish the security of Protocol \podtm-core shown Construction~\ref{const:pod-core}, we first prove some useful intermediate results. We remind that $n = \alpha + \beta + \gamma$, where $n$ denotes the total number of replicas, $\beta$ denotes the number of Byzantine replicas, \gamma denotes the number of omission-faulty replicas in an execution, and \alpha denotes the number of replicas required to confirm a transaction.

\begin{lemma}[The values for minimum, maximum and \confirmed rounds]\label{lem:assign-min-max}
	Regarding \Cref{alg:client-statistics}, we have the following.
	Consider the list of all timestamps received by a client for a particular transaction, replacing a missing vote from $\rep_j$ with a special value ($\mrt[\rep_j]$ for computing \rmin, $\infty$ for computing \rmax), to get $n$ values in total, sorted in increasing order.
	Assume \mrt is also sorted in increasing order of timestamps.
	\begin{enumerate}
		\item \rmin is the timestamp at index $\lfloor \alpha / 2 \rfloor - \beta$ of this list.
		\item \rmax is the timestamp at index $n - \alpha + \lfloor \alpha / 2 \rfloor + \beta$ of this list.
		\item \rperf is the timestamp at index $\lfloor \alpha / 2 \rfloor - \beta$ of \mrt.
	\end{enumerate}
\end{lemma}
\begin{proof}
	Functions \opMinPossibleTs{} and \opMinPossibleTsForNewTx{} prepend \beta times the $0$ value in the beginning of the list and return the median of the first \alpha values, hence they return the timestamp at index $\lfloor \alpha / 2 \rfloor - \beta$.
	Function \opMaxPossibleTs{} appends \beta times the $\infty$ value at the end of the list and returns the median of the last \alpha values of that list, that is, it ignores the first $n - \alpha  + \beta$ values and returns the timestamp at index $n - \alpha  + \beta + \lfloor \alpha / 2 \rfloor$.
\end{proof}

\begin{lemma}[\rperf bounded by honest timestamp]\label{lem:rperf-value-of-honest}
	Assuming $n \geq 5 \beta + 3 \gamma + 1$ (equiv., $\alpha \geq 4\beta + 2 \gamma + 1$),
	for a valid \Pod with auxiliary data $\cert = (\ppcert, \alltxcert)$,
	there exists some honest replica $\rep_j$, such that the most-recent timestamp \mrt from $\rep_j$ included in $\ppcert$ satisfies $\mrt \leq \Pod.\rperf$.
\end{lemma}
\begin{proof}
	Since $\opValid{\Pod, \cert} = \true$, the past-perfect round $\Pod.\rperf$ is the value returned by \opMinPossibleTsForNewTx{} of \Cref{alg:client-statistics}.
	From \Cref{lem:assign-min-max} we have that \rperf is the timestamp at index $\lfloor \alpha / 2 \rfloor - \beta$ of sorted \mrt.
	The condition $\alpha \geq 4\beta + 2 \gamma + 1$ implies that $\beta + \gamma \leq \lfloor \alpha / 2 \rfloor - \beta$,
	hence the number of not honest replicas ($\beta + \gamma$) cannot fill all positions between $0$ and $\lfloor \alpha / 2 \rfloor - \beta$,
	hence at least one of the indexes between $0$ and $\lfloor \alpha / 2 \rfloor - \beta$ (inclusive) will
	contain the timestamp created and sent by an honest replica.
\end{proof}

\noindent

We now recall Theorem~\ref{thm:pod-security}, which we prove through a series of lemmas.

\begin{customthm}{\ref{thm:pod-security}}[\podtm-core security]
	Assume that the network is partially synchronous with actual network delay $\delta$, that \beta is the number of Byzantine replicas, \gamma the number of omission-faulty replicas, $\alpha = n - \beta - \gamma$ the confirmation threshold, and $n \geq 5 \beta + 3 \gamma + 1$ the total number of replicas. Protocol \podtm-core (\Cref{const:pod-core}), instantiated with a EUF-CMA secure signature scheme, the \opValid{} function shown in \Cref{alg:pod-core-valid}, and the \identify{} function described in \Cref{alg:identify}, is a responsive secure \podtm (\Cref{def:podsec}) with Confirmation within $\parConf = 2 \delta$, Past-perfection within $\parPerf = \delta$ and  \beta-accountable safety (\Cref{def:accountable-safety}), \ewnp.
\end{customthm}
\begin{proof}
	The proof follows from Lemmas~\ref{lem:pod-confirmation}--\ref{lem:pod-accountable-safety}, presented and proven in the remainder of this section.
\end{proof}

\begin{lemma}[Confirmation within $\parConf$]\label{lem:pod-confirmation}
	For the conditions stated in \Cref{thm:pod-security}, \Cref{const:pod-core} satisfies the \emph{confirmation within $\parConf$} property (\Cref{def:podsec}) for $\parConf = 2 \delta$.
\end{lemma}
\begin{proof}
	Assume an honest client \client calls \opwrite{\tx} at round \round.
	It sends a message \msgwrite{\tx} to all replicas at round \round (line~\ref{line:send-write-tx}).
	An honest replica receives this by round $\round + \delta$ and sends a \msgrecord{} message back to all connected clients (line~\ref{line:send-vote}).
	An honest client $\client'$ receives the vote by round $\round + 2\delta$.
	As are at least \alpha honest (not Byzantine and not omission-faulty) replicas, $\client'$ receives at least \alpha such votes, hence the condition in line~\ref{line:confirm-tx-begin} is satisfied and $\client'$ observes \tx as \confirmed.
\end{proof}

\begin{lemma}[Past-perfection within $\parPerf$]\label{lem:past-perfection-liveness}
	For the conditions stated in \Cref{thm:pod-security}, \Cref{const:pod-core} satisfies the \emph{past-perfection within $\parPerf$} property (\Cref{def:podsec}) for $\parPerf = \delta$.
\end{lemma}
\begin{proof}
	Assume an honest client \client at round \round has view $\podView{c}{\round}$.
	From \Cref{lem:rperf-value-of-honest}, there exists some honest replica $\rep_j$, such that the most-recent timestamp $\mrt[\rep_j]$ that $\rep_j$ has sent to \client satisfies $\podView{c}{\round}.\rperf \geq \mrt[\rep_j]$.
	The honest replica $\rep_j$ sends at least one heartbeat or vote message per round (line~\ref{line:rep-send-vote-heartbeat}), which arrives within $\delta$ rounds, and an honest client updates $\mrt[\rep_j]$ when it receives the heartbeat or vote message.
	Hence, \client will have $\mrt[\rep_j] \geq \round - \delta$.
	All together, $\podView{c}{\round}.\rperf \geq \round - \delta$.
\end{proof}

\begin{lemma}[Past-perfection safety]\label{lem:pod-pp-safety}
	For the conditions stated in \Cref{thm:pod-security}, \Cref{const:pod-core} satisfies the \emph{past-perfection} safety property (\Cref{def:podsec}), except with negligible probability.
\end{lemma}
\begin{proof}
	Assume the adversary outputs valid $(\Pod_1, \cert_1)$ and $(\Pod_2, \cert_2)$ that violate the property,
	i.e., there exists a transaction $\tx$ such that $(\tx, \rmin^1, \rmax^1, \rconf^1) \not\in \Pod_1.\txSet$ and $(\tx, \rmin^2, \rmax^2, \rconf^2) \in \Pod_2.\txSet$ and $\rconf^2 \neq \bot$ and $\rconf^2 < \Pod_1.\rperf$.
	Let $\cert_1 = (\ppcert^1, \alltxcert^1)$ and $\cert_2 = (\ppcert^2, \alltxcert^2)$.

	Let $\allreps_1$ be the set of replicas $\rep_i$ for which $\ppcert^1$ contains a vote with timestamp $\mrt_i \geq \Pod_1.\rperf$.
	From \Cref{lem:assign-min-max} (\rperf is computed as the timestamp at index $\lfloor \alpha / 2 \rfloor - \beta$ of sorted \mrt), and since $\Pod_1$ is valid, there exist at least $n - \lfloor a/2 \rfloor + \beta$ such replicas, hence $|\allreps_1| \geq n - \lfloor a/2 \rfloor + \beta$.
	For each $\rep_i \in \allreps_1$, the transaction certificates $\alltxcert^1$ contain the whole log of $\rep_i$ with timestamps up to $\mrt_i$ (\cref{line:client-check-sn} of \Cref{alg:client-1} does not allow gaps in the sequence number of the received votes).
	That is, for each $\rep_i \in \allreps_1$ the certificates $\alltxcert^1$ contains votes
	\begin{equation}\label{eq:proof-pp-set1}
		(\tx_{i,1}, \tsp_{i,1}, 1, \sigma_{i,1}, \rep_i), (\tx_{i,2}, \tsp_{i,2}, 2, \sigma_{i,2}, \rep_i), \ldots, (\tx_{i,k_i}, \tsp_{i, k_i}, k_i, \sigma_{i, k_i}, \rep_i),
	\end{equation}
	where $k_i$ is the smallest sequence number for which $\tsp_{i, k_i} \geq \Pod_1.\rperf$, and $\tx_{i,j}$ are transactions.

	Since \tx is confirmed in $\Pod_2$ and $\rconf^2 < \Pod_1.\rperf$, the transaction certificate $\alltxcert^2[\tx]$ must contain votes on \tx with timestamp $\tsp_i$, such that $\tsp_i < \Pod_1.\rperf$, from at least $\lfloor \alpha/2 \rfloor + 1$ replicas.
	Let $\allreps_2$ be the set of these replicas, with $|\allreps_2| \geq \lfloor \alpha/2 \rfloor + 1$.
	For each $\rep_i \in \allreps_2$, certificate $\alltxcert^2[\tx]$ contains a vote
	\begin{equation}\label{eq:proof-pp-vote2}
		(\tx, \tsp_i, \sn_i, \sigma_i, \rep_i),
	\end{equation}
	such that $\tsp_i < \Pod_1.\rperf$.
	We will show that, if at most \beta replicas are Byzantine, this leads to a contradiction.
	Observe from the cardinality of $\allreps_1$ and $\allreps_2$ that at least $\beta + 1$ replicas must be in both sets, hence at least one honest replica must be in both sets (except if the adversary forges a signature under the public key of an honest replica, which happens with negligible probability).
	For that replica, the vote in (\ref{eq:proof-pp-vote2}) must be one of the votes in (\ref{eq:proof-pp-set1}) since $\tsp_i < \Pod_1.\rperf$ and $\tsp_{i, m_i} \geq \Pod_1.\rperf$.
	Hence, one of the $\tx_{i,j}$ in (\ref{eq:proof-pp-set1}) is \tx, and $\tx$ must appear in $\Pod_1.\txSet$, a contradiction.
\end{proof}


\begin{lemma}[Confirmation bounds]\label{lem:bounds-safety}
	For the conditions stated in \Cref{thm:pod-security}, \Cref{const:pod-core} satisfies the \emph{confirmation bounds} safety property (\Cref{def:podsec}), except with negligible probability.
\end{lemma}
\begin{proof}
	Assume the adversary outputs $(\Pod_1, \cert_1)$ and $(\Pod_2, \cert_2)$, such that $\opValid{\Pod_1, \cert_1} \land \opValid{\Pod_2, \cert_2}$ and there exists a transaction $\tx$ such that $(\tx, \rmin^1, \rmax^1, \rconf^1) \in \Pod_1.\txSet$ and $(\tx, \rmin^2, \rmax^2, \rconf^2) \in \Pod_2.\txSet$.
	Let $\cert_1 = (\ppcert^1, \alltxcert^1)$ and $\cert_2 = (\ppcert^2, \alltxcert^2)$, and $\txcert^1 = \alltxcert^1[\tx]$ and $\txcert^2 = \alltxcert^2[\tx]$.

	First assume $\rmin^1 > \rconf^2$.
	From \Cref{lem:assign-min-max}, $\txcert^1$ can include at most $\lfloor \alpha / 2 \rfloor - \beta$ votes with a timestamp for \tx smaller than $\rmin^1$.
	Allowing up to \beta replicas to equivocate, the adversary can obtain at most $\lfloor \alpha / 2 \rfloor$ votes on \tx with a timestamp smaller than $\rmin^1$, except if it forges a digital signature from an honest replica, which happens with negligible probability.
	In order to compute $\rconf^2 < \rmin^1$ for \tx, the adversary must include in $\txcert^2$ timestamps smaller than $\rmin^1$ from at least $\lfloor \alpha / 2 \rfloor + 1$ replicas.

	Now assume $\rmax^1 < \rconf^2$.
	Using \Cref{lem:assign-min-max}, $\txcert^1$ can include at most $\alpha - \lfloor \alpha / 2 \rfloor - \beta - 1$ votes with a timestamp larger than \rmax,
	hence the number of honest replicas, from which a vote with timestamp larger than \rmax can be included in $\txcert^2$ is at most $\alpha - \lfloor \alpha / 2 \rfloor - 1$ (since \beta are malicious).
	If \alpha is odd, this upper bound becomes $\alpha - \lfloor \alpha / 2 \rfloor -1 = \lfloor \alpha / 2 \rfloor $, while at least
	$\lfloor \alpha / 2 \rfloor + 1$ votes larger that \rmax are required to compute a median larger than \rmax,
	and if \alpha is even, then $\alpha - \lfloor \alpha / 2 \rfloor - 1 = \lfloor \alpha / 2 \rfloor - 1$, while at least
	$\lfloor \alpha / 2 \rfloor$ votes larger that \rmax are required to compute a median larger than \rmax.
	(we remind that \cref{alg:client-statistics} returns as median the value at position $\lfloor \alpha /  2 \rfloor$).
	In either case, we get a contradiction, except for the negligible probability that the adversary forges a digital signature from an honest replica.
\end{proof}

\begin{algorithm}[]
    \caption{The \identify{} function for Protocol \podtm-core (\Cref{const:pod-core}).}
    \label{alg:identify}
    \begin{algorithmic}[1]
        \Function{\identify{\tran}}{}
            \Let{\cheaters}{\emptyset}
            \For{$ \msgrecord{(\tx_1, \tsp_1, \sn_1, \sig_1, \rep_1)} \in \tran$}
                \If{\textbf{not} \opVerify{$\pk_1, (\tx_1, \tsp_1, \sn_1), \sig_1$}}
                    \State{\textbf{continue}}
                \EndIf
                \For{$ \msgrecord{(\tx_2, \tsp_2, \sn_2, \sig_2, \rep_2)} \in \tran$}
                    \If{\textbf{not} \opVerify{$\pk_2, (\tx_2, \tsp_2, \sn_2), \sig_2$}}
                        \State{\textbf{continue}}
                    \EndIf
                    \If{$\rep_1  = \rep_2 \textbf { and } \sn_1 = \sn_2 \textbf { and } (\tx_1 \neq \tx_2 \textbf { or } \tsp_1 \neq \tsp_2)$ }\label{line:identify-cheat}
                        \Let{\cheaters}{\cheaters \cup \{\rep_1\}}
                    \EndIf
                \EndFor
            \EndFor
        \EndFunction
    \end{algorithmic}
    \end{algorithm}

\begin{lemma}[\beta-Accountable safety]\label{lem:pod-accountable-safety}
	For the conditions stated in \Cref{thm:pod-security}, \Cref{const:pod-core} satisfies accountable safety (\Cref{def:accountable-safety}) with resilience \beta, except with negligible probability.
\end{lemma}
\begin{proof}
	We show that \identify{} (\Cref{alg:identify}) satisfies the \emph{correctness} and \emph{no-framing} properties required by \Cref{def:accountable-safety}, in three steps.

	\dotparagraph{1}
	If the past-perfection safety property (\Cref{def:podsec}) is violated, there exists a partial transcript \tran, such that \identify{} on input \tran returns at least \beta replicas.

	\noindent
	\emph{Proof:}
	We resume the proof of \Cref{lem:pod-pp-safety}. There, we constructed sets $\allreps_1, \allreps_2$,
	such that $\allreps_1 \cap \allreps_2 \geq \beta + 1$.
	We saw that, for each $\rep_i \in \allreps_1 \cap \allreps_2$, certificates $\alltxcert^1$ contain the replica log shown in (\ref{eq:proof-pp-set1}), containing all votes with timestamp up to $\tsp_{i, k_i} \geq \rperf$.
	In a similar logic, certificates $\alltxcert^2$ contains the following $k'_i$ votes from $\rep_i$ (possibly more, but we care for the votes up to transaction \tx)
	\begin{equation}\label{eq:proof-pp-set2}
		(\tx'_{i,1}, \tsp'_{i,1}, 1, \sigma'_{i,1}, \rep_i), (\tx'_{i,2}, \tsp'_{i,2}, 2, \sigma'_{i,2}, \rep_i), \ldots, (\tx'_{i, k'_i}, \tsp'_{i, k'_i}, k'_i, \sigma'_{i, k'_i}, \rep_i),
	\end{equation}
	with $\tx'_{i, k'_i} = \tx$ and $\tsp'_{i, k'_i} < \rperf$.
	Obviously, for an honest $\rep_i$, the replica logs of (\ref{eq:proof-pp-set1}) and (\ref{eq:proof-pp-set2}) must be identical, i.e., $\tx_{i,j} = \tx'_{i,j}$ and $\tsp_{i,j} = \tsp'_{i,j}$, for $j \in [ 1, \min(k_i, k'_i)]$.
	We will show that they differ in at least one sequence number.
	If $k_i > k'_i$, then the replica logs differ at sequence number $k'_i$, because the transaction $\tx_{i, k_i}$ in (\ref{eq:proof-pp-set1}) cannot be \tx, as $\Pod_1.\txSet$ does not contain \tx, and $\tx'_{i, k'_i} = \tx$.
	If $k_i \leq k'_i$, the log of (\ref{eq:proof-pp-set1}) should be identical with the first $k_i$ positions of the log of (\ref{eq:proof-pp-set2}), which would imply that $\tsp_{i, k_i} = \tsp'_{i, k_i}$ and, since a valid pod only accepts non-decreasing timestamps, $\tsp'_{i, k_i} \leq \tsp'_{i, k'_i}$, and all together $\tsp_{i, k_i} \leq \tsp'_{i, k'_i}$.
	This is impossible, because $\tsp_{i, k_i} > \rperf$ and $\tsp'_{i, k'_i} < \rperf$.
	Hence, the two logs will contain a different timestamp for some sequence number in $[1, k'_i]$.

	Summarizing, we have shown for at least $\beta + 1$ replicas $\rep_i \in \allreps_1 \cap \allreps_2$, certificate $\cert_1$ and $\cert_2$ contain votes $(\tx_1, \tsp_1, \sn_1, \sig_1, \rep_i)$ and $(\tx_2, \tsp_2, \sn_2, \sig_2, \rep_i)$, such that $\sn_1 = \sn_2$ but $\tx_1 \neq \tx_2$ or $\tsp_1 \neq \tsp_2$.
	On input a set $\tran$ that contains these votes, function \identify{\tran} returns $\allreps_1 \cap \allreps_2$.

	\dotparagraph{2}
	If the confirmation-bounds property (\Cref{def:podsec}) is violated, there exists a partial transcript \tran, such that \Cref{alg:identify} on input \tran returns at least \beta replicas.

	\noindent
	\emph{Proof:}
	As in the proof of \Cref{lem:bounds-safety}, assume the adversary outputs $(\Pod_1, \cert_1)$ and $(\Pod_2, \cert_2)$, such that $\opValid{\Pod_1, \cert_1} \land \opValid{\Pod_2, \cert_2}$ and there exists a transaction $\tx$ such that $(\tx, \rmin^1, \rmax^1, \rconf^1) \in \Pod_1.\txSet$, $(\tx, \rmin^2, \rmax^2, \rconf^2) \in \Pod_2.\txSet$, and $\rmin^1 > \rconf^2 \lor \rmax^1 < \rconf^2$
	Let $\cert_1 = (\ppcert^1, \alltxcert^1)$ and $\cert_2 = (\ppcert^2, \alltxcert^2)$, and $\txcert^1 = \alltxcert^1[\tx]$ and $\txcert^2 = \alltxcert^2[\tx]$.

	Let's take the case $\rmin^1 > \rconf^2$ first. From \Cref{lem:assign-min-max} (\txVotesMin contains at least $n - \lfloor \alpha / 2 \rfloor + \beta$ timestamps \tsp such that $\tsp \geq \rmin$), there is a set $\allreps_1$ with at least $n - \lfloor \alpha / 2 \rfloor + \beta$ replicas $\rep_i$, from each of which $\alltxcert^1$ contains votes
	\begin{equation}\label{eq:proof-bounds-set1}
		(\tx_{i,1}, \tsp_{i,1}, 1, \sigma_{i,1}, \rep_i), (\tx_{i,2}, \tsp_{i,2}, 2, \sigma_{i,2}, \rep_i),
		\ldots, (\tx_{i,m_i}, \tsp_{i, m_i}, m_i, \sigma_{i, m_i}, \rep_i),
	\end{equation}
	up to some sequence number $m_i$, such that $\tsp_{i, m_i} \geq \rmin$ and either $\tx_{i,m_i} = \tx$ (i.e., a vote from $\rep_i$ on \tx is included in $\txcert^1$, and we only consider the votes up to this one), or  $\tx_{i, j} \neq \tx, \forall j \leq m_i$ (i.e., a vote from $\rep_i$ on \tx is not included in $\txcert^1$, in which case \txVotesMin contains the timestamp $\rep_i$ has sent on $\tx_{i,m_i} \neq \tx$).

	Now, for a valid $\Pod_2$ to output $\rconf^2 < \rmin^1$, certificate $\txcert^2$ must contain timestamps smaller than \rmin from at least $\lfloor \alpha / 2 \rfloor + 1$ replicas. Call this set $\allreps_2$. From each of these replicas, certificates $\alltxcert^2$ must contain votes
	\begin{equation}\label{eq:proof-bounds-set2}
		(\tx'_{i,1}, \tsp'_{i,1}, 1, \sigma'_{i,1}, \rep_i), (\tx'_{i,2}, \tsp'_{i,2}, 2, \sigma'_{i,2}, \rep_i), \ldots, (\tx, \tsp'_{i, m'_i}, m'_i, \sigma'_{i, m'_i}, \rep_i),
	\end{equation}
	considering only votes up to \tx, for which $\tsp'_{i, m'_i} < \rmin$.

	By counting arguments there are at least $\beta + 1$ replicas in $\allreps_1 \cap \allreps_2$.
	For each one, we make the following argument. Since $\tsp_{i, m_i} \geq \rmin$ and $\tsp'_{i, m'_i} < \rmin$, we get $\tsp'_{i, m'_i} < \tsp_{i, m_i}$, and it must be the case that $m'_i < m_i$ (otherwise, the two logs will differ at a smaller sequence number, similar to the previous case).
	But in this case the two logs differ at sequence number $m'_i$, i.e., $\tx_{i,m'_i} \neq \tx'_{i,m'_i} = \tx$.
	This is because the log of (\ref{eq:proof-bounds-set1}) either does not contain \tx, or contains it at sequence number $m_i > m'_i$, in which case it must contain a different transaction at sequence number $m'_i$.
	On input a set $\tran$ that contains all votes for replicas in $\allreps_1$ and $\allreps_2$ votes, function \identify{\tran} returns $\allreps_1 \cap \allreps_2$.

	For the case $\rmax^1 < \rconf^2$, similar arguments apply.
	In order to compute $\rconf^2 > \rmax^1$, certificate $\txcert^2$ must contain at least $\lfloor \alpha / 2 \rfloor$ or $\lfloor \alpha / 2 \rfloor + 1$ (depending on the parity of $\alpha$) votes on \tx with timestamp larger than \rmax.
	On the other hand, from \Cref{lem:assign-min-max} certificate $\txcert^1$ contains at least $n - \alpha + \lfloor \alpha / 2 \rfloor + \beta$ votes on \tx with a timestamp smaller or equal than \rmax.
	As before, the replicas in the intersection of these two sets have sent conflicting votes for some sequence numbers.

	\dotparagraph{3}
	The \identify{} function never outputs honest replicas.

	\noindent
	\emph{Proof:} The function only adds a replica to \cheaters if given as input two vote messages from that replica, where the same sequence number is assigned to two different votes (line~\ref{line:identify-cheat} on \Cref{alg:identify}). An honest replica always increments \nextsn after each vote it inserts to its log (line~\ref{line:rep-increment-sn} on \Cref{alg:pod-core-replica}), hence, the adversary can only construct such verifying votes by forging a signature under the public key of an honest replica, which happens with negligible probability.
\end{proof}
\section{Proofs for additional \podtm properties}
In this section we prove the $\parTemp$-timeliness property for \podtm, as stated in~\Cref{app:monotonicity}.

\begin{theorem}[$\parTemp$-timeliness for honest transactions]
	\sloppy{
		For the conditions stated in \Cref{thm:pod-security}, \Cref{const:pod-core} satisfies \emph{\parTemp-timeliness for honest transactions} (\Cref{def:pod-timely}), for $\parTemp = \delta$, except with negligible probability.
	}
\end{theorem}
\begin{proof}
	Assume an honest client \client calls \opwrite{\tx} at round \round.
	It sends a message \msgwrite{\tx} to all replicas at round \round (line~\ref{line:send-write-tx}).
	An honest replica receives this by round $\round + \delta$ and assigns its current round, which lies in the interval $(\round, \round + \delta]$, as the timestamp (line~\ref{line:rep-assign-tsp}).
	\begin{enumerate}
		\item Regarding \rconf, when a client calls \opreadall{} (after the point in time when \tx is confirmed, which happens after \parConf rounds from the property of \emph{confirmation within \parConf}), it receives votes on \tx from at least \alpha replicas. All honest replicas have sent timestamps for \tx in the interval $ (r, \round + \delta]$. Since \rconf is computed as the median of \alpha timestamps and $\alpha \geq 4\beta + 2\gamma + 1$,\footnote{For this argument on \rconf, $\alpha \geq 2\beta + 2\gamma + 1$ would also be enough. The condition $\alpha \geq 4\beta + 2\gamma + 1$ is necessary in order for \rmin and \rmax of a confirmed transaction to be timestamps returned by honest replicas.}
		we get $\lfloor \alpha / 2 \rfloor > 2\beta + \gamma$,
		hence \rconf will be a timestamp returned by an honest (not Byzantine and not omitting messages) replica, or it will lie between timestamps returned by honest replicas. Hence, $\rconf \in (r, \round + \delta]$.

		\item Regarding \rmax, from \Cref{lem:assign-min-max} (\rmax is the timestamp at index $n - \alpha + \lfloor \alpha / 2 \rfloor + \beta$ of \txVotes), there areat least $\alpha - \lfloor \alpha/2 \rfloor - \beta + 1 > \lfloor \alpha/2 \rfloor - \beta$ timestamps in \txVotes that bound \rmax from above. Since $\alpha \geq 4\beta + 2 \gamma + 1$, we get that $\lfloor \alpha / 2 \rfloor > 2\beta + \gamma$, hence $\lfloor \alpha / 2 \rfloor - \beta > \beta + \gamma$, hence at least one of those timestamps that bound \rmax is returned by an honest replica, hence  $\rmax \in (r, \round + \delta]$.

		\item Similarly, from \Cref{lem:assign-min-max} (\rmin is the timestamp at index $\lfloor \alpha / 2 \rfloor - \beta$ of \txVotes), there are at least $\lfloor \alpha / 2 \rfloor - \beta + 1$ timestamps in \txVotes that bound \rmin from below, and, since $\alpha \geq 4\beta + 2 \gamma + 1$, we get $\lfloor \alpha / 2 \rfloor - \beta + 1 > \beta + \gamma$.
		Hence, \rmin is a timestamp returned by an honest replica, hence $\rmin \in (r, \round + \delta]$ and $\rmax - \rmin < \parTemp$.
	\end{enumerate}
	The proofs hold except with negligible probability, as the adversary can forge a signature under the public key of an honest replica with a negligible probability.
\end{proof}

\section{Security of \bidsettm-core}\label{app:bidsetsec}



In this section, we recall and prove \Cref{thm:bidset-construction-secure}.

\begin{customthm}{\ref{thm:bidset-construction-secure}}[Bidset security]
    Assuming a synchronous network where $\delta \leq \Delta$,
    protocol \bidsettm-core (Construction~\ref{const:bidset}) instantiated with a digital signature and a secure \podtm protocol that satisfies the \emph{past-perfection within $\parPerf = \delta$}, \emph{confirmation within $\parConf=2\delta$} and \emph{$\delta$-timeliness} properties, is a secure \bidsettm protocol satisfying \emph{termination within $\parAuctionTerm = 3\Delta + \delta$}.
    It satisfies accountable safety with an \identifySequencer{} function that identifies a malicious sequencer.
\end{customthm}
\begin{proof}
	In Lemmas \ref{lem:bidset-termination}--\ref{lem:bidset-accountability}. The function for identifying a malicious sequencer is shown in \Cref{alg:identify-bidset}.
\end{proof}

\begin{lemma}[Termination within \parAuctionTerm]\label{lem:bidset-termination}
    Under the assumptions of \Cref{thm:bidset-construction-secure}, \Cref{const:bidset} satisfies \emph{termination within $\parAuctionTerm = \tzero + 3\Delta + \delta$}.
\end{lemma}
\begin{proof}
    The \opResult{} event is generated by an honest consumer when its exits the loop of lines~\ref{line:bidset-wait-pp-consumer}--\ref{line:bidset-consumer-loop-end} in \Cref{alg:bid-set-consumer}. At the latest, this happens when round $\tzero + 3\Delta$ becomes past-perfect (\cref{line:bidset-read-result-pp} in \Cref{alg:bid-set-consumer}), which, from the \emph{past-perfection within $\delta$} property of \podtm, happens at round at most $\tzero + 3\Delta + \delta$, hence $\parAuctionTerm = \tzero + 3\Delta + \delta$.
    We remark that a sequencer (\Cref{alg:bid-set-sequencer}) also terminates, because from the \emph{past-perfection within $\delta$} property of \podtm, the condition of line~\ref{line:bidset-wait-pp} becomes true by round $\tzero + \Delta + \delta$.
\end{proof}

\begin{lemma}[Censorship resistance]\label{lem:bidset-censorship}
    Under the assumptions of \Cref{thm:bidset-construction-secure}, \Cref{const:bidset} satisfies the \emph{censorship resistance} property.
\end{lemma}
\begin{proof}
    Assume the sequencer is honest, and an honest bidder calls \opSubmitBid{\bid} at time \tzero. We will show that $\bid \in \bag$.
    First, the \podtm view $\podView{\auc}{\round}$ of the sequencer \auc on the round \round when it constructs \bag satisfies $\podView{\auc}{\round}.\rperf > \ttwo$.
    Second, from the \emph{confirmation within \parConf} property of \podtm, the transaction containing \bid becomes confirmed, and from the \emph{\parTemp-timeliness} property of \podtm, it gets a confirmation round $\rconf \leq \tzero + \parTemp$. For $\parTemp = \delta$, and since $\delta \leq \Delta$, we get that $\rconf \leq \ttwo$.
    Hence, from the past-perfection safety property of \podtm we get that $\bid \in \podView{\auc}{\round}$, and, since the sequencer is honest, $\bid \in \bag$.
\end{proof}

\begin{lemma}[Consistency]\label{lem:bidset-consistencys}
    Under the assumptions of \Cref{thm:bidset-construction-secure}, \Cref{const:bidset} satisfies the \emph{consistency} property.
\end{lemma}
\begin{proof}
    Assume the sequencer is honest, and two honest consumers generate events \opResult{$\bag_1, \cdot$} and \opResult{$\bag_2, \cdot$}.
    The condition in \cref{line:bidset-wait-pp} of \Cref{alg:bid-set-sequencer} becomes true in the view of sequencer by round $\tzero + \Delta + \delta$ (from the \emph{past-perfection within $\parPerf=\delta$} property of \podtm-core),
    hence the sequencer writes transaction \msgbidresult{(\bag, \bidcert, \sigma)} to \podtm by round $\tzero + \Delta + \delta$.
    This transaction gets assigned a confirmed round $\rconf \leq \tzero + \Delta + 2\delta$ (from the \emph{$\delta$-timeliness} property of \podtm) and,
    by assumption of a synchronous network, $\rconf \leq \tzero + 3 \Delta$.
    The condition in \cref{line:bidset-find-confirmed-tx} of \Cref{alg:bid-set-consumer} requires that a round $\round' > \tzero + 3 \Delta$ becomes past perfect.
    As $\round' > \rconf$, and by \emph{past-perfection safety} of \podtm, the consumer observes the transaction as confirmed before $\round'$ becomes past-perfect, hence the condition in \cref{line:bidset-find-confirmed-tx} becomes true before the condition in \cref{line:bidset-read-result-pp} and an honest consumer outputs $\opResult{\bag, \bidcert}$.
\end{proof}

\begin{lemma}[Accountable safety]\label{lem:bidset-accountability}
    Under the assumptions of \Cref{thm:bidset-construction-secure}, and assuming that \podInst is an instance of \podtm-core, \Cref{const:bidset} achieves \emph{accountable safety}, using the \identifySequencer{} function (\Cref{alg:identify-bidset}) to identify a malicious sequencer.
\end{lemma}
\begin{proof}
    Following \Cref{sec:accountable-safety}, we show an \identifySequencer{\tran} function (\Cref{alg:identify-bidset}), that, on input a partial transcript \tran outputs $\true$ when safety is violated due to misbehavior of the sequencer (i.e., it identifies the sequencer as malicious), and $\false$ if the sequencer is honest.
    We prove the theorem in three parts.
    \dotparagraph{1} For violations of censorship-resistance:

    \noindent
    Assume an honest bidder calls \opSubmitBid{\bid} at time \tzero and the network is synchronous.
    The transaction \tx containing \bid becomes confirmed, and any honest party can observe $(\tx, \rconf, \cdot, \cdot)$ and the corresponding transaction certificate \txcert in their view of the \podtm, as returned by \podtm-core.
    Assume \bid is censored, i.e., an event \opResult{\bag, \bidcert} is output by an honest consumer, such that $\bid \not \in \bag$.
    Let \sig be the signature of the sequencer in the corresponding \msgbidresult{(\bag, \bidcert, \sigma)} message written on \podInst.
    We will show how the sequencer can be made accountable, using $(\txcert, \bag, \bidcert, \sigma)$ as evidence \tran.
    In order for $(\txcert, \bag, \bidcert, \sigma)$ to be valid evidence, the following must hold:

    \emph{Requirement 1}: The signature \sig must be a valid signature, produced by the sequencer on message $(\bag, \bidcert)$, as per \cref{line:bidset-sequencer-sign} of the sequencer code (checked on \cref{line:bidset-identify-check-sig} of \Cref{alg:identify-bidset} -- we remind that notation `\textbf{require} $P$' returns $\false$ if $P$ evaluates to $\false$).


    \emph{Requirement 2}: \txcert must contain at least \alpha votes (checked on \cref{line:bidset-identify-check-large-txcert} of \Cref{alg:identify-bidset}), on the same transaction $\tx^*$ (checked on \cref{line:bidset-identify-check-same-tx}), signed by a \podInst replica (checked on \cref{line:bidset-identify-check-votes-txcert}).

    If any of these requirements are not met, \tran does not constitute valid evidence and the function exits. Otherwise, let $\rconf^*$ be the median of all votes in \txcert. The function makes the following checks, and if any of them fails, then the sequencer is accountable.

    \emph{Check 1}: Verify whether the votes that the sequencer has included in \bidcert are valid, obtained from the replicas that run \podInst (lines~\ref{line:bidset-identify-votes-begin}-\ref{line:bidset-identify-votes-end}). If this is not the case, the sequencer has misbehaved.

    \emph{Check 2}: Compute the \rperf from the timestamps found in the votes in \bidcert (lines~\ref{line:bidset-identify-rperf-begin}-\ref{line:bidset-identify-rperf-end}). This \rperf must be larger than $\tzero + \Delta$, as per \cref{line:bidset-wait-pp} of \Cref{alg:bid-set-sequencer}.

    \emph{Check 3}: If $\rconf^* \leq \tzero+\Delta$ but $\tx^*$ is not in the bag, the sequencer has misbehaved.

    \dotparagraph{2} For violations of consistency:

    \noindent
    The consistency property can be violated if the sequencer writes two transactions $\msgbidresult{(\bag_1, \cdot, \cdot)}$ and $\msgbidresult{(\bag_2, \cdot,\cdot)}$ to \podInst, such that $\bag_1 \neq \bag_2$, in which case $\bag_1$ and $\bag_2$ identify the sequencer. As this is a simpler case, we do not show it in \Cref{alg:identify-bidset}.

    \dotparagraph{3} A honest sequencer cannot be framed:

    \noindent
    Finally, we show that an honest sequencer cannot be framed.
    If the sequencer has followed \Cref{alg:bid-set-sequencer}, then \bidcert will contain valid votes, hence \emph{Check 1} will pass.
    Moreover, an honest sequencer waits until the past-perfect round returned by the \podtm is larger than $\tzero + \Delta$, hence \emph{Check 2} will pass.
    Regarding \emph{Check 3}, observe that
    for \identifySequencer{} to compute $\rconf^* \leq \tzero+\Delta$, \txcert must contain at least $\lfloor \alpha/2 \rfloor + 1$ votes on $\tx^*$ with a timestamp smaller or equal than $\tzero+\Delta$, and at least $\lfloor \alpha/2 \rfloor + 1 - \beta$ of them must be from honest replicas. Call this set $\allreps'$.
    The honest sequencer, in order to output a past-perfect round greater than $\tzero+\Delta$, must have received timestamps greater than $\tzero+\Delta$ from at least $n - \lfloor a/2 \rfloor + \beta$ replicas (from \Cref{lem:assign-min-max}).
    By counting arguments, at least one of these timestamps must be from one of the honest replicas in $\allreps'$, and, since honest replicas do not omit transactions, that replica will have sent a vote on $\tx^*$ to the sequencer. Hence, the honest sequencer will include $\tx^*$ in \bag.
\end{proof}

\begin{algorithm}
    \caption{The \identifySequencer{} function for \Cref{const:bidset}, instantiated with an instance of \podtm-core (\Cref{const:pod-core}) as \podInst, run by a set of replicas $\allreps = \{\rep_1, \cdots, \rep_n\}$ with public keys $\{\pk1, \cdots, \pk_n \}$, and using the \opMed{} operation defined by \Cref{const:pod-core}. It identifies a malicious sequencer, whose public key is $\pk_a$, by returning $\true$.}
    \label{alg:identify-bidset}
    \begin{algorithmic}[1]
        \Function{\identify{\tran}}{}
            \Let{(\txcert, \bag, \bidcert, \sigma)}{\tran}
            \State{\textbf{require} \opVerify{$\pk_a, (\bag, \bidcert), \sigma$}} \label{line:bidset-identify-check-sig}
            \Let{\txVotesPerf}{[ \, ]}
            \For{$\vote \in \bidcert$} \label{line:bidset-identify-votes-begin}
                \Let{(\tx, \tsp, \sn, \sig, \rep_j)}{\vote}
                \If{\opVerify{$\pk_j, (\tx, \tsp, \sn), \sig$ = 0}}
                    \State{\textbf{return true}} \label{line:bidset-identify-check-votes-ppcert}
                \EndIf
                \Let{\txVotesPerf}{\txVotesPerf \concat \tsp}
            \EndFor{}\label{line:bidset-identify-votes-end}

            \State{\text{sort \txVotesPerf in increasing order}} \label{line:bidset-identify-rperf-begin}
            \Let{\txVotesPerf}{ [ 0, \stackrel{\beta \text{ times}}{\ldots}, 0 ] \concat \txVotesPerf}
            \Let{\rperf}{\opMed{$\txVotesPerf[:\alpha]$}}
            \If{$\rperf \leq \tzero + \Delta$}
                \State{\textbf{return} $\true$}
            \EndIf \label{line:bidset-identify-rperf-end}

            \Statex{}
            \State{\textbf{require} $|\txcert| \geq \alpha$} \label{line:bidset-identify-check-large-txcert}
            \Let{\tx^*}{\txcert[0].\tx}
            \Let{\txVotes}{[ \, ]}
            \For{$\vote \in \txcert$}
                \Let{(\tx, \tsp, \sn, \sig, \rep_j)}{\vote}
                \State{\textbf{require} $\tx = \tx^*$} \label{line:bidset-identify-check-same-tx}
                \State{\textbf{require} \opVerify{$\pk_j, (\tx, \tsp, \sn), \sig$}} \label{line:bidset-identify-check-votes-txcert}
                \Let{\txVotesPerf}{\txVotesPerf \concat \tsp}
            \EndFor{}
            \Let{\rconf^*}{\opMed{\txVotes}}
            \If{$\rconf^* \leq \tzero+ \Delta \textbf{ and } \tx^* \not \in \bag$}
                \State{\textbf{return} $\true$}
            \EndIf \label{line:bidset-identify-inclusion-end}
        \EndFunction
\end{algorithmic}
\end{algorithm}

\end{document}